%% file: header.tex
\documentclass[sigconf]{acmart}

\usepackage{booktabs} % For formal tables
\usepackage{url}
\usepackage{color}
\usepackage{hyperref}
\hypersetup{
     colorlinks   = true,
     citecolor    = gray,
     linkcolor    = gray,
     urlcolor     = gray
}
\usepackage[shortlabels]{enumitem}
\setlist[enumerate]{leftmargin=*}
\setlist[itemize]{leftmargin=*}

\usepackage{graphicx}
\usepackage{verbatim} % for comments
\graphicspath{ {images/} }

\usepackage{listings}
\newcommand{\dol}{\mbox{\textdollar}} % for $ in FCF
\usepackage [english]{babel}
\usepackage [autostyle, english = american]{csquotes}
\MakeOuterQuote{"}

\usepackage{amssymb}
\usepackage{amsmath}
\usepackage{amsthm}
\usepackage{dcolumn}
\usepackage{hyperref}
\usepackage{enumitem}  % For the new environment with args

\usepackage{enumitem}
\newlist{longenum}{enumerate}{6}
\setlist[longenum,1]{label=\arabic*.}
\setlist[longenum,2]{label=\alph*.}
\setlist[longenum,3]{label=\roman*.}
\setlist[longenum,4]{label=(\arabic*)}
\setlist[longenum,5]{label=(\alph*)}
\setlist[longenum,6]{label=(\roman*)}

\newenvironment{game}
{ \begin{itemize}[noitemsep,nolistsep] 
}
{ \end{itemize}                  }

\newcommand{\s} {\textrm{ }}
\newcommand{\bv} {\begin{verbatim}}
\newcommand{\ev} {\end{verbatim}}

\newcommand{\li} {\lstinline}
\newcommand{\kv} {$(k, v)$}

\newcommand{\lar}{\leftarrow}

\copyrightyear{2017}
\acmYear{2017}
\setcopyright{acmcopyright}
\acmConference{CCS '17}{}{Oct. 30--Nov. 3, 2017, Dallas, TX, USA.}
\acmPrice{15.00}
\acmDOI{10.1145/3133956.3133974}
\acmISBN{ISBN 978-1-4503-4946-8/17/10}

\newcommand\BLINDFORSUBMISSION[2]{#2}

\begin{document}

\title{Verified Correctness and Security of mbedTLS HMAC-DRBG}
% \titlenote{Produces the permission block, and
%   copyright information}
% \subtitle{Extended Abstract}
% \subtitlenote{The full version of the author's guide is available as
  % \texttt{acmart.pdf} document}

\BLINDFORSUBMISSION{
  \author{Authors anonymized for submission to ACM CCS 2017}
  }{
\author{Katherine Q. Ye}
\affiliation{
  \institution{Princeton U., Carnegie Mellon U.}
}

\author{Matthew Green}
\affiliation{
  \institution{Johns Hopkins University}
}

\author{Naphat Sanguansin}
\affiliation{
  \institution{Princeton University}
}

\author{Lennart Beringer}
\affiliation{
  \institution{Princeton University}
}

\author{Adam Petcher}
\affiliation{
  \institution{Oracle}
}

\author{Andrew W. Appel}
\affiliation{
  \institution{Princeton University}
}
}

\newcommand{\mdg}[1]{}
\newcommand{\awa}[1]{}
\newcommand{\kye}[1]{}
\newcommand{\lenb}[1]{}

\begin{abstract}
% \kye{Andrew: add a sentence introducing the problem with today's DRBGs}

  We have formalized the functional specification of HMAC-DRBG (NIST 800-90A), and we have proved its cryptographic security---that its output is pseudorandom---using a hybrid game-based proof. We have also proved that the mbedTLS implementation (C program) correctly implements this functional specification. That proof composes with an existing C compiler correctness proof to guarantee, end-to-end, that the machine language program gives strong pseudorandomness. All proofs (hybrid games, C program verification, compiler, and their composition) are machine-checked in the Coq proof assistant. Our proofs are modular: the hybrid game proof holds on any implementation of HMAC-DRBG that satisfies our functional specification. Therefore, our functional specification can serve as a high-assurance reference.

  % \kye{Andrew: add a sentence on implications/significance}
\end{abstract}

%
% The code below should be generated by the tool at
% http://dl.acm.org/ccs.cfm
% Please copy and paste the code instead of the example below. 
%
% \begin{CCSXML}
% <ccs2012>
%  <concept>
%   <concept_id>10010520.10010553.10010562</concept_id>
%   <concept_desc>Computer systems organization~Embedded systems</concept_desc>
%   <concept_significance>500</concept_significance>
%  </concept>
%  <concept>
%   <concept_id>10010520.10010575.10010755</concept_id>
%   <concept_desc>Computer systems organization~Redundancy</concept_desc>
%   <concept_significance>300</concept_significance>
%  </concept>
%  <concept>
%   <concept_id>10010520.10010553.10010554</concept_id>
%   <concept_desc>Computer systems organization~Robotics</concept_desc>
%   <concept_significance>100</concept_significance>
%  </concept>
%  <concept>
%   <concept_id>10003033.10003083.10003095</concept_id>
%   <concept_desc>Networks~Network reliability</concept_desc>
%   <concept_significance>100</concept_significance>
%  </concept>
% </ccs2012>  
% \end{CCSXML}

% \ccsdesc[500]{Computer systems organization~Embedded systems}
% \ccsdesc[300]{Computer systems organization~Redundancy}
% \ccsdesc{Computer systems organization~Robotics}
% \ccsdesc[100]{Networks~Network reliability}

% We no longer use \terms command
%\terms{Theory}

% \keywords{}

\maketitle

\input{paperbody}

\paragraph{Acknowledgments.}
This research was supported in part by DARPA agreement FA8750-12-2-0293
and by NSF Grant CCF-1521602.
{\small
The U.S. Government is authorized
to reproduce and distribute reprints for Governmental purposes
notwithstanding any copyright notation thereon. The views and
conclusions contained herein are those of the authors and should not
be interpreted as necessarily representing the official policies or
endorsements, either expressed or implied, of DARPA or the
U.S. Government.}

\bibliographystyle{ACM-Reference-Format}
\bibliography{references} 

\input{appendix}

\end{document}

%% file: paperbody.tex
% !TEX root = header.tex
% Page limit is 12 pages
% Intro: 2 pages
% Security proof: 2 pages
% Mechanization: 2 pages
% Correctness proof: 2 page
% Spec equivalence: 1.5 pages
% Proof effort: 0.5 page
% Related work/future work/conclusion: 1 page
% References: 0.5 pages
% Total: 11.5 pages

\section{Introduction}
\label{sec:intro}
Cryptographic systems require large amounts of randomness to generate keys, nonces and initialization vectors. Because many computers lack the large amount of high-quality physical randomness needed to generate these values, most cryptographic devices rely on pseudorandom generators (also known as {\em deterministic random bit generators} or DRBGs) to ``stretch'' small amounts of true randomness into large amounts of pseudorandom output.\footnote{A note on terminology: we use ``entropy'' loosely to denote randomness that is not predictable by an adversary. We use ``sampled uniformly at random'' and ``ideally random'' interchangeably. We use PRG, the acronym for ``pseudo-random generator,'' to refer to the abstract cryptographic concept, whereas we use DRBG, the acronym for ``deterministic random bit generator,'' to denote the specifications and implementations of PRGs. Instead of DRBG, some papers use ``PRNG,'' the acronym for ``pseudo-random number generator.'' The terms are synonymous.} 

Pseudorandom generators are crucial to security. Compromising a generator (for example, by selecting malicious algorithm constants, or exfiltrating generator state) can harm the security of nearly any cryptosystem built on top of it. The harm can be catastrophic: for example, an adversary who can predict future outputs of the generator may be able to predict private keys that will be generated, or recover long term keys used as input to the protocol execution.

Moreover, it may be impossible to test that a DRBG has been compromised if one is limited to black-box testing. Current validation standards \cite{fips140-2,iso19790,NISTRNGToolkit} primarily resort to statistical tests and test vectors, neither of which guarantee that the output is pseudorandom.
One can even construct backdoored PRGs that cannot be detected by black-box testing \cite{dodis2015formal}.

Despite the importance of DRBGs, their development has not received the scrutiny it deserves. Many catastrophic flaws in DRBGs have been uncovered at both the design level and the implementation level. Some bugs arise from simple programming mistakes. The Debian DRBG, though open-source, was broken for two years because a line of code was erroneously removed, weakening the DRBG's seeding process~\cite{debianflaw}. Several recent projects have identified and factored large numbers of weak RSA keys produced due to improper generator seeding~\cite{Heninger12,Lenstra12,Bernstein13}. A bug in entropy use in the Android DRBG resulted in the theft of \$5,700 of Bitcoin \cite{goodin13}.

While those flaws were accidental, some are malicious. Adversaries intentionally target DRBGs because breaking a DRBG is an easy way to break the larger cryptosystem. The most notorious example is the NSA's alleged backdooring of the Dual EC DRBG standard~\cite{sp800-90A,Shumow07,NYT_DualEC_2013,DBLP:conf/birthday/BernsteinLN16}. In the Dual EC design, the malicious choice of a single algorithm parameter (an elliptic curve point $Q$) allows for full state recovery of the DRBG given only a small amount of raw generator output. This enables the passive decryption of protocols such as TLS and IPSEC/IKE~\cite{Shumow07,Checkoway14,Checkoway16}. From 2012 to 2015 this backdoor was exploited by unauthorized parties to insert a passive decryption backdoor into Juniper NetScreen VPN devices~\cite{Wired_Juniper_2015,Checkoway16}. Remarkably, a related DRBG flaw in Fortinet VPNs created a similar vulnerability in those devices during the same time period~\cite{FortinetBug}.

A key weakness in the deployment of DRBGs is that current government standards both encourage specific designs and lack rigor. The FIPS validation process (required by the U.S. government for certain types of cryptographic device) mandates the use of an ``approved'' NIST DRBG in every cryptographic module. As a consequence, a small number of DRBG designs have become ubiquitous throughout the industry. These algorithms lack formal security proofs, and their design processes were neither open nor rigorous---in fact, the designs were found to include errors. Even worse, it is easy to implement certain of these DRBGs such that the output of the generator is predictable (given detailed knowledge of the implementation), and yet without this knowledge {\em the output of the generator is computationally indistinguishable from random} \cite{dodis2015formal}. Unfortunately, the existing formal validation processes for verifying the correctness of DRBG implementations are weak and routinely ignore entire classes of flaws. Given that the widespread deployment of a weak DRBG can undermine the security of entire computer networks, we need a better way to validate these critical systems.

% Should the NIST/FIPS-specific angle be interleaved throughout the intro?

\subsection{Contributions}

DRBGs have the special property that testing, even sophisticated statistical fuzz-testing, cannot assure the security of an implementation. Even at the level of a pure cryptographic protocol, proving the security of a DRBG construction can be quite challenging. (Hybrid-game-based proofs are the only known technique, but for nontrivial DRBGs those proofs have so many technical steps that it is hard to trust them.) Therefore a new paradigm is required, in which the functional model of a DRBG is proved \emph{with a machine-checked proof} to satisfy the appropriate PRF properties, and the C language implementation (\emph{and} its compilation to machine language) is proved \emph{with a machine-checked proof} to correctly implement the functional model, and these proofs are linked together in the same proof-checker.

We present machine-checked proofs, in Coq, of many components, connected and checked at their specification interfaces so that we get a truly end-to-end result:
\textit{Verson 2.1.1 of the mbedTLS HMAC-DRBG correctly implements the NIST 800-90A standard, and HMAC-DRBG \textbf{Generate} and \textbf{Update} as described in that same NIST 800-90A standard indeed produces pseudorandom output, subject to the standard assumptions\footnote{The standard PRF assumptions about the security of SHA-2 are stated by Beringer \emph{et al.}\cite[\S 4.2]{beringer15:hmac}. No one knows how to prove these, but it is not controversial to assume them. See also \S\ref{concrete} of our paper.}
 about SHA-2, as well as certain assumptions about the adversary and the HMAC-DRBG instantiation that we state formally and explicitly. We have not proved the security of \textbf{Instantiate} (see \S\ref{thebad})
and \textbf{Reseed}.} 

To prove this theorem, we took the following steps.

\begin{enumerate}
\item We constructed a proof that the version of HMAC-DRBG described in the standard NIST SP 800-90A generates pseudorandom output, assuming that HMAC is a pseudorandom function and that the DRBG's state is instantiated with ideal randomness. Our proof is similar to that of Hirose \cite{hirose2008security}, but more clearly and explicitly structured. By ``pseudorandom,'' we mean that a nonadaptive\footnote{In many typical applications of DRBGs, such as TLS, adversaries are nonadaptive because they are not given control over parameters such as the number of requests to the DRBG or the block length.} probabilistic polynomial-time adversary that requests the maximum amount of output from HMAC-DRBG cannot distinguish its output from uniformly random output with nonnegligible probability.

\item We mechanized the proof in (1),
  formalizing the cryptographic algorithm \emph{(``crypto spec'')} and the
  security theorem in Coq, instantiating HMAC as HMAC/SHA-256,
  and assuming that the DRBG's initial state is ideally random.
Our crypto spec was written with respect to the NIST standard SP~800-90A \cite{sp800-90A} which defines HMAC-DRBG. It was written in the probabilistic programming language provided by the Foundational Cryptography Framework \cite{petcher15:post}, which is itself embedded in the Coq proof assistant.

\item We proved that the mbedTLS implementation of HMAC-DRBG (written in C)
  correctly implements a functional specification, also derived
  from SP~800-90A,   written as a pure function in Coq.
  We performed the verification using the Verified Software Toolchain (VST) framework \cite{appel14:plcc}, which is also embedded in the Coq proof assistant. 

\item We proved that our two functional specs for HMAC-DRBG
  (both derived from SP~800-90A) are equivalent,
  by induction in Coq. This connects the C-program correctness
  proof to the cryptographic pseudorandomness proof.

\item Beringer \emph{et al.} \cite{beringer15:hmac}
  and Appel \cite{appel15:sha}
  proved that an OpenSSL implementation of HMAC with SHA-256
  implements the FIPS-198 and FIPS-180 functional specs (respectively),
  and those in turn implement a PRF with bounded attacker advantage, subject
  to some standard (unproven but uncontroversial) assumptions about the
  security of SHA-256's compression function.

\end{enumerate}

Composing these proofs (which we do in the Coq proof assistant without any gaps at the interfaces) yields our main result.
This modular approach allows one to incrementally verify even complex cryptographic protocol implementations in real languages.

\paragraph{The trusted base.}  Appel \cite[\S 12]{appel15:sha} discusses at length
the trusted code base of this approach.
In summary: the mbedTLS DRBG program need not be trusted, as
it is proved to implement the spec; the compiler need not be trusted,
as it is proved correct; the NIST spec (or our formalization of it)
need not be trusted, because it is proved to be a PRF.
FCF need not be trusted, because it is proved sound in Coq (only
the \emph{definitions} used in the statement of the PRF theorem need
be trusted, because if you prove the wrong theorem it's not so useful).
VST need not be trusted, because it is proved sound w.r.t. the
operational semantics of CompCert C.  The CompCert verified optimizing
C compiler need not be trusted, because it is proved correct w.r.t.
that same operational semantics and the op. sem. of the
assembly language.
The op. sem. specification of CompCert's
source language (C) need not be trusted, because it's just an internal
interface of the proof (between VST and CompCert).
 What must be trusted are: the specification
of the assembly language ISA (as an operational-semantic statement in Coq);
the statement of the main theorem and the definitions it relies on
 (one must make sure that the
right property is proved); and the Coq kernel (proof checker),
along with the OCaml compiler that compiled Coq, and the OCaml
compiler that compiled the OCaml compiler (it's turtles all
the way down \cite{thompson1984reflections}). One could also compile mbedTLS with gcc or clang, and in that case
the compiler would be part of the TCB.

Our DRBG proof is for mbedTLS, and our HMAC proof is for OpenSSL,
so technically our end-to-end proof is for an implementation with
those two components stapled together. \kye{Anonymize ``our''}
% The connection of proofs is
% quite modular, so if were to port our HMAC proof to mbedTLS-HMAC---or
% even if we replaced HMAC with another cryptographic PRF---the DRBG proof would%  not have to be touched.

We prove ``only'' functional correctness, which has useful corollaries
of safety and direct information flow: no buffer overruns, no reading
from or writing to unspecified locations.  We do not address side
channels---which are not a functional-correctness property---but
because we verify standard C-language implementations, any other
analysis tools for side channels in C programs can be applied. Closer
to machine code, our approach is compatible with Barthe \emph{et
  al.}'s methodology of combining functional and implementational
correctness proofs with formal proofs of leakage security, based on a
noninterference analysis for CompCert's low-level intermediate
language Mach~\cite{DBLP:conf/ccs/BartheBCLP14}.

\paragraph{Coq proofs.}
For the functional specification, cryptographic specification, proof of
\label{coq-proofs}
functional correctness, and security proof, see \newline
\textit{github.com/PrincetonUniversity/VST/tree/master/hmacdrbg}.
\linebreak
The README file in that directory
explains how the Coq files correspond to sections of this paper.

\subsection{Prior work}

Hirose \cite{hirose2008security} proves the pseudorandomness of the full HMAC-DRBG on paper, also with respect to a nonadaptive adversary. However, the proof is not linked to any implementation of HMAC-DRBG. Moreover, though it is well-written, the six-page Hirose proof is challenging to read and verify. For example, Hirose justifies his arguments informally; several times, he states that an important implication is ``easy to see'' or that it ``directly follows from the preceding lemmas.'' Hirose also asserts that HMAC-DRBG is secure when the user provides an optional additional input on initialization; however, he only states the relevant lemmas without proving them. We do not prove any properties relating to additional input.

To perform proofs more readably and rigorously on complex constructions such as DRBG, we harness the power of mechanization. We mechanize our proof in the Coq proof assistant, using the embedded Foundational Cryptography Framework (FCF). In contrast to working on paper, mechanizing a proof requires that all lemmas and relationships between lemmas be precisely stated. Additionally, FCF encourages proofs to be done in the game-playing style, which makes the intermediate definitions more explicit (because they must all be written as code) and the proof as a whole easier to read and verify. See \S \ref{sec:fullproof} for an outline of the structure of our proof.

More specifically, the Coq proof assistant assists the user by providing an interactive development environment and user-prog\-ram\-mable automation to construct a formal syntactic proof. Then, the Coq kernel checks that this term proves the stated theorem. Coq has been used to prove mathematical theorems, cryptographic theorems, C program correctness, and C compiler correctness. The latter three are relevant to our work.

Affeldt et al. \cite{affeldt2012certifying} present an end-to-end
machine-checked (in Coq) proof of the Blum-Blum-Shub PRG, proving
cryptographic security of the algorithm and correctness of the
implementation.  They mechanize an existing, well-known proof of BBS's security.
Their cryptographic games are linked
directly to their own assembly code, whereas our games apply to a high-level
specification of the DRBG in question, and connect to a widely
used C program.  BBS is not used
in practice because it is quite slow \cite{sidorenko05}.

Appel \cite{appel15:sha} proved in Coq that OpenSSL SHA-256
correctly implements the FIPS-180 specification.
Beringer \emph{et al.} \cite{beringer15:hmac}
proved in Coq
the PRF security of the FIPS-198 specification of
HMAC/SHA-256, and that the OpenSSL implementation of
HMAC correctly implements FIPS-198.  

Petcher \cite{petcher15:phd} applies mechanization to demonstrate a \textit{negative} example, rather than a positive one. He demonstrates how the process of mechanizing the proof of security of Dual-EC-DRBG naturally results in an unprovable theorem, which reveals the weakness in the DRBG's design.
Our mechanized proof of HMAC-DRBG shows which parts have no flaws
(\emph{Generate} and \emph{Update}), and which are problematic (\emph{Instantiate}).

Inspired by the Android bug described earlier, the Entroscope tool
\cite{DorreKlebanov2016} uses information flow and static analysis
(with bounded model checking) to detect entropy loss in the
seeding of a PRNG.  Entroscope has been successfully applied to find a
bug in the Libgcrypt PRNG, and would have detected the Debian incident~\cite{debianflaw}.

Higher-level protocols such as TLS have been
proved secure with machine-checked proofs \cite{bhargavan2013implementing},
although without proving security of primitives such as DRBGs or HMAC.
Our work complements such proofs.

No prior work has proved full cryptographic security and correctness of a DRBG that links all of its components: a specification on paper, a proof of security on paper, a widely used implementation, and its correct
compilation to assembly code.
The DRBG flaws listed earlier demonstrate that a DRBG may contain flaws anywhere (and everywhere!) in that stack. Thus, it is imperative to verify DRBGs, and cryptographic primitives in general, in a manner that links all of their components.

\section{Deterministic Random Bit Generators}
\label{sec:drbgsec}

\newcommand{\instantiate}{\textbf{Instantiate}}
\newcommand{\update}{\textbf{Update}}
\newcommand{\reseed}{\textbf{Reseed}}
\newcommand{\generate}{\textbf{Generate}}
\newcommand{\negl}{\mathbf{negl}}
\newcommand{\prg}{{\sf PR}}
\newcommand{\AlgA}{A}
\newcommand{\compind}{\stackrel{c}{\approx}}
\newcommand{\pushin}{\hspace{.15in}}

In 2007 the U.S. National Institute of Standards and Technology (NIST) released Special Publication 800-90 describing several pseudorandom generator constructions, replaced in 2012 by 800-90A ~\cite{sp800-90A}.
Termed {\em deterministic random bit generators} in NIST's parlance, these generators have been widely adopted in many commercial products, including many operating systems and TLS libraries. This adoption is largely because SP800-90A generators are required in modules validated by the U.S. government's Cryptographic Module Validation Program~\cite{fips140-2}. 

The NIST specification describes a specific algorithmic interface as well as an informal list of security goals. We present this below.

\noindent
\begin{definition}
A DRBG is a tuple of algorithms $(\instantiate, \allowbreak \update, \allowbreak \reseed, \allowbreak \generate)$ with the following interface.

\begin{description}
\item[] $\instantiate(entropy, nonce).$ On input of an initial entropy string and a personalization nonce, outputs the initial state $\langle{k,v}\rangle$.
\item[] $\update(\mathit{data}, \langle{k, v}\rangle).$ On input of (optional) \emph{data} string (for seeding) and a generator state $\langle{k,v}\rangle$, outputs a new state $\langle{k', v'}\rangle$.
\item[] $\reseed(\langle{k, v}\rangle, entropy).$ On input of the current generator state and an entropy string, output a new state $\langle{k',v'}\rangle$.
\item[] $\generate(\langle{k, v}\rangle, n).$ On input of the current generator state and a number of blocks $n$, output a string of pseudorandom output as well as a new state $\langle{k,v}\rangle$.
\end{description}
\end{definition}

\noindent
The NIST standard informally describes several security properties that a DRBG must possess. These include the standard notion of pseudorandomness, as well as several auxiliary properties. In this work we focus primarily on the pseudorandomness property.% which we define below.

\begin{definition}[Pseudorandomness of the DRBG] \label{def:pseudorandomness}
\label{def:cpa}
~\newline
Let $\Pi= (\instantiate, \allowbreak \update, \allowbreak \reseed, \allowbreak \generate)$ be a DRBG, let $c$ denote the length of the internal state vector, and let the random variable
$\prg_b(\Pi,\AlgA,\lambda)$ where $b\in \{0,1\}$,
and $\lambda \in \mathbb{N}$ denote the result of the
following probabilistic experiment:

\begin{center}
\begin{tabular}{l@{\hspace{.3in}}l}
    ${\prg}_b(\Pi,\AlgA,\lambda)~~:=$ &
      \pushin $E \gets \{0,1\}^{\lambda}$ \\
&    \pushin $(k,v) \gets \instantiate(E, \bot)$ \\
&    \pushin $B \gets \AlgA^{{\mathcal O}_{{\sf Generate}}(b, \cdot)}(\lambda)$\\
&    \pushin Output $B$
\end{tabular}
\end{center}
\noindent where ${\mathcal O}_{\sf Generate}(b, n)$ defines an oracle that is initialized with initial state $(k,v)$ and if $b=0$ runs the $\generate$ algorithm using the internal state and then updates the internal state.
The security parameter $\lambda$ is the number of entropy bits
seeding the DRBG.
If $b=1$, the oracle samples a uniformly random string of length $nc$ (where $c$ is the block length) and returns this to $\AlgA$. A DRBG scheme $\Pi$ is pseudorandom if $\forall$ p.p.t.\ algorithms $\AlgA$ the following two ensembles are computationally indistinguishable:
$$ \Bigl\{ \prg_0(\Pi,\AlgA,\lambda) \Bigr\}_{\lambda}
 \ \compind \
   \Bigl\{ \prg_1(\Pi,\AlgA,\lambda) \Bigr\}_{\lambda}
$$
\end{definition}

In this work we will focus primarily on the $\generate$ and $\update$ algorithms, while leaving analysis of the $\instantiate$ algorithm for future analysis (we discuss the reasons in Section~\ref{goodbadugly}). Thus in the remainder of this work we will employ the simplifying assumption that the initial state $(k,v)$ is uniformly random.

\subsection{HMAC-DRBG} \label{sec:hmac_drbg_design}

HMAC-DRBG is one of the four pseudorandom generator constructions formalized in NIST SP 800-90A. It works by iterating HMAC, a keyed-hash message authentication function that can be instantiated using a cryptographic hash function such as SHA. 
HMAC is useful in constructing a PRG. When instantiated with an efficient cryptographic hash function it is highly efficient. Bellare~\cite[\S5]{bellare2006new} shows that HMAC is a {\em pseudorandom function} if certain assumptions hold on the hash functions's internal compression function.

HMAC-DRBG has an internal state consisting of two pieces of administrative information, which are constant across calls, and two pieces of secret information, which are updated during each call to the PRG.
The administrative component of the state contains the security strength of the PRG's instantiation, for which the \emph{Instantiate} function obtains the appropriate amount of entropy for that security strength, and a flag that indicates whether this instantiation requires prediction resistance. 

The secret component of the internal state is the working state \kv{},
where $k$ is a secret key of length $c$ for the internal PRF (e.g., HMAC). Each time the PRG generates a string of pseudorandom bits (by calling the $\generate$ function), $k$ is updated at least once. The $v$ component holds the newest ``block'' of pseudorandom bits that the PRG has generated. The length of $v$ is the length of HMAC's output, which is fixed (e.g., 256 bits).
% TODO fix informal "block" terminology

% TODO comment on getEntropy
% HMAC-DRBG consists of four functions, $\update$, $\instantiate$, $\reseed$ and $\generate$, and another function called \emph{getEntropy} that we do not model. 

Here we will rephrase the NIST standard in a mostly functional style that closely mirrors our specification in Coq.
Let $f$ denote HMAC and $f_k$ denote HMAC partially applied with the key $k$. Let $||$ denote concatenation. \\

\emph{\update} refreshes $(k, v)$, optionally using some $\mathit{data}$, which can be fresh entropy for prediction resistance.\\

\emph{\instantiate} initializes $(k,v)$ and \emph{reseed\_counter}.  For simplicity, we omit the \emph{additional\_input} and \emph{personalization\_string} parameters from \emph{Instantiate}
the ``crypto spec'' that we present here;
these optional parameters allow for the provision of additional entropy to the instantiate call.
Our proof of the mbedTLS C code (see \S\ref{sec:mbedtlsVerification}) does
handle these parameters.
The omitted parameters are generally used for fork-safety, so therefore
we have not proved fork-safety.
\pagebreak

\noindent $\textbf{Update}(\mathit{data}, k, v) := $
\begin{game}
\item[] $k_1 \lar f_k(v || 0x00 || data)$
\item[] $v_1 \lar f_{k_1}(v)$
\item[] if $(\mathit{data} = \mathrm{nil})$
\item[] then return $(k_1, v_1)$
\item[] else $k_2 \lar f_{k_1}(v_1 || 0x01 || \mathit{data})$
\item[] \qquad $v_2 \lar f_{k_2}(v_1)$
\item[] \qquad return $ \s (k_2, v_2)$ \\
\end{game}

\noindent $\textbf{Instantiate}(entropy, nonce) := $
\begin{game}
\item[] $\mathit{seed} \lar \mathit{entropy} || \mathit{nonce}$
\item[] $k \lar 0x00 \ldots 00$
\item[] $v \lar 0x01 \ldots 01$
\item[] $(k_1, v_1) \lar \mathbf{Update}(\mathit{seed}, k, v)$
\item[] $\mathit{reseed\_counter} \lar 1$
\item[] return $\s (k_1, v_1, \mathit{reseed\_counter})$
\end{game}

$\reseed$ could be called for two reasons. First, \emph{reseed\_counter} may have exceeded \emph{reseed\_interval}. This is rare, since \emph{reseed\_interval} is set to high values in practice; its maximum value is $2^{48}$, meaning it would naturally reseed once every couple of million years. (Indeed, in practice, one does not want a PRG to reseed often. This would give an attacker more opportunities to compromise the entropy source.) More commonly, $\reseed$ is called when a PRG's state could have been compromised, and requires fresh entropy to be mixed in.\\ % We acknowledge this is a conflict in goals for a DRBG. 

% \emph{Generate} generates $n$ pseudorandom bits.\\
% \emph{Gen\_loop} is the internal bit-generating loop of \emph{Generate}, rephrased in a functional and recursive manner from the original imperative \emph{while} loop. $c$ is the input/output size of $f_k$ (HMAC). \\

% $\textbf{Gen\_loop(n, k, v)} :=$
% \begin{game}
% \item[] if $n = 0$ then return nil (the empty list)
% \item[] else return $f_k(v) || H(n-1, k, f_k(v)).$\\
% \end{game}

\noindent $\textbf{Generate}(k, v, n) := $\qquad %\emph{(* Generate n pseudorandom bits *)}
\begin{game}
 \item[] (* We don't model the reseed test.
 \item[] if $\mathit{reseed\_counter} > \mathit{reseed\_interval}$ then \ldots reseed \ldots
%   \item[] $bits \lar Gen\_loop(\lceil \f{n}{c} \rceil, k,v)$
 \item[] else *)
\item[] \qquad $\mathit{temp}_0 \lar nil$
\item[] \qquad $v_0 \lar v$
\item[] \qquad $i \lar 0$
\item[] \qquad while $\mathrm{length}(\mathit{temp}_i) < n$
\item[] \qquad ~do~ $v_{i+1} \lar f_k(v_i)$
\item[] \qquad\qquad $\mathit{temp}_{i+1} \lar \mathit{temp}_i || v_{i+1}$
\item[] \qquad\qquad  $i \lar i+1$
\item[] \qquad $\mathit{returned\_bits} \lar$ leftmost $n$ of $\mathit{temp}_i$
\item[] \qquad $(k_1, v_{i+1}) \lar \mathbf{Update}(k, v_i)$
\item[] \qquad $\mathit{reseed\_counter}_1 \lar \mathit{reseed\_counter} + 1$
\item[] \qquad return  $ \s (bits, k_1, v_{i+1}, \mathit{reseed\_counter}_1)$\\
\end{game}

\noindent 
$\textbf{Reseed}(k, v, entropy) := $
\begin{game}
\item[] $(k_1, v_1) \lar Update(seed, k, v)$
\item[] $reseed\_counter \lar 1$
\item[] return $ \s (k,v, reseed\_counter)$ \\
\end{game}

\noindent
From the user's point of view, only the \emph{Instantiate} and \emph{Generate} functions are visible. The state is hidden. Typical usage consists of a call to \emph{Instantiate}, followed by any number of calls to \emph{Generate} to produce pseudorandom bits. \emph{Generate} automatically calls \emph{Update} every time it is called, and \emph{Generate} may force a call to \emph{Reseed}.

% We discuss our proof of security for a simplified PRG, consisting of only one call to the internal loop in \emph{Generate}, in Section \ref{sec:indistinguishability_proof_for_one_call}. We discuss our proof of security for a simplified version of HMAC-DRBG in Section \ref{sec:indistinguishability_proof_for_a_simplified_hmac_drbg}.

\pagebreak
\section{Overview of the proof}

Here we explain the techniques used in our proof.
Section~\ref{sec:mechanized} presents the mechanized proof in detail.

\subsection{Key result: HMAC-DRBG is pseudorandom} \label{maintheorem}
We prove the pseudorandomness property of HMAC-DRBG. As with traditional proofs, we show the computational indistinguishability of two main ``experiments.'' In the first (or ``real'') experiment, which roughly corresponds to the experiment $\prg_0$ in \S\ref{sec:drbgsec}, the adversary interacts with an oracle whose generator uses HMAC-DRBG. In the second (or ``ideal'') experiment, the adversary interacts with an oracle that produces uniformly random output. This corresponds to $\prg_1$. Our goal is to connect these experiments via a series of intermediate ``hybrid'' experiments, which we will use to show that the real and ideal experiments are computationally indistinguishable by any probabilistic polynomial time (p.p.t.) adversary.

\begin{description}
\item {\bf The real experiment ($G_\mathrm{real}$).} In this experiment, a p.p.t. adversary $A$ makes a sequence of calls to an oracle that implements the $\generate$ function of DRBG. At the start of the experiment, the generator is seeded with a uniformly random generator state $\langle{k,v}\rangle$, and at each call the adversary requests a fixed number $n$ of blocks of generator output, where $n$ is chosen in advance. The oracle returns the requested output, then updates its internal state as required by the $\generate$ algorithm. At the conclusion of the experiment, the adversary produces an output bit $b$.

\item {\bf The ideal experiment ($G_\mathrm{ideal}$).} The adversary $A$ interacts with an oracle according to the same interface as in the previous experiment. However, this oracle simply returns a uniformly random string of length $n$ blocks. The oracle maintains no internal state. At the conclusion of the experiment, the adversary produces an output bit $b$.
\end{description}

\noindent
Stated informally, the probability that $A$ outputs $1$ when interacting with $G_\mathrm{real}$, minus the probability
that the adversary outputs $1$ when interacting with $G_\mathrm{ideal}$, is a measure of how successful $A$ is in distinguishing the pseudorandom
generator from a random generator. Our goal is to bound this quantity to be at most a negligible function of the generator's security parameter (see \S\ref{concrete}).

Our proof connects the {\em real} and {\em ideal} experiments by defining a series of intermediate, or hybrid, experiments. These hybrid experiments comprise a transition between the {\em real} and {\em ideal} experiments. We show that for each successive pair of hybrids---beginning with the real experiment and ending with the ideal experiment---the distribution of the adversary's output must be at most negligibly different between the two experiments in the pair. Having performed this task for each pair of experiments leading from {\em real} to {\em ideal}, we then apply a  standard hybrid argument to prove that if all intermediate hybrids satisfy our condition, then the real and ideal experiments themselves must also be computationally indistinguishable. It remains now to outline our hybrids. 

\medskip \noindent 
{\em Overview of the hybrids.} Recall that the HMAC-DRBG generation procedure consists of an iterated construction with two main stages. Given a uniformly random internal state $(k,v_0)$ the first phase of generation operates as follows: a call to the $\generate$ function produces a sequence of output blocks which are calculated by applying the HMAC function $f_k(v_0)$ to generate a block of output $v_1$, and then applying $f_k (v_1)$ to obtain $v_2$ and so on until $n$ blocks have been generated. The second phase of $\generate$ {updates} the internal state of the generator to produce a new state $(k',v'_0)$ which will be used in the subsequent call to $\generate$.

% Thus our proof can be broken into two main components: $(1)$ a first set of hybrids that shows that {\em at each call to $\generate$}, given a polynomial integer $n$ and a uniformly random initial state $(k,v_0)$, the output of $\generate(k,v,n)$ is indistinguishable from a random string. 
% And $(2)$ a second argument that shows that the newly updated state $(k',v'_0)$ produced at the end of each call (and used as input to the subsequent call) is itself indistinguishable from random. By combining these two arguments 

Our proof defines a set of hybrids that shows that {\em at each call to $\generate$}, given integer $n$ (whose size is polynomial in $\lambda$, the number
of entropy bits)
and a uniformly random initial state $(k,v_0)$, the output of $\generate(k,v,n)$ is indistinguishable from a random string. Using these hybrids, we show that the PRG, over all calls to $\generate$, is secure for any fixed number of calls \emph{numCalls} made by the adversary. 

We define the following sequence of hybrids:

\begin{description}
\item$G_\mathrm{real}$. This is the real experiment in which all calls are answered using the HMAC-DRBG $\generate$ algorithm.

\item $G_{1\dots \mathit{numCalls}}$. For each hybrid indexed by $i \in [1, \ldots, \mathit{numCalls}]$, we modify the experiment such that the first $i$ calls are replaced with a function that simply returns ideally random bits and a new uniformly random state $(k',v'_0)$ when queried on any input.
%For $i=0$ to \emph{numCalls} we modify the output of each call $j \le i$ as follows:
Internally this is achieved using two additional hybrids:
\begin{enumerate}
\item {\bf Internal hybrid 1.} In this hybrid, the first $i-1$ oracle calls produce random output. In call $i$ we modify the output to sample a random function $r$ and use $r$ in place of the PRF $f$ to compute the $\generate$ algorithm.
\item {\bf Internal hybrid 2.} Identical to the previous hybrid, except that we substitute each output of the random function with a string of uniformly random bits.
% \item {\bf Internal hybrid 3.} Identical to the previous hybrid, except that we substitute the output state $(k',v'_0)$ with a uniformly random output.
\end{enumerate}
\item $G_\mathrm{ideal}$. This is the ideal experiment.
\end{description}

In the first hybrid, we replace the HMAC function $f$ with a random function $r$. We are able to show that this replacement does not affect $A$'s output if $f$ is pseudorandom, as this would imply the existence of a distinguisher for the pseudorandom function, a fact which would directly contradict our assumption. With this replacement complete, hybrid $(2)$ requires us to make specific arguments about the input to the random function $r$, namely that the inputs to the random function do not collide. Thus, the core of our proof is a sequence of theorems arguing that, from the perspective of $A$, the changes we make at each of the intermediate hybrids described above induces only a negligible difference in $A$'s output. We conclude then by making a standard hybrid argument that bounds the overall distance of $G_\mathrm{real}$ and $G_\mathrm{ideal}$ to be negligible in $\lambda$. We calculate a concrete bound on $A$'s ability to distinguish between $G_\mathrm{real}$ and $G_\mathrm{real}$.

% deleted a bunch of commented-out stuff here, moved to file ``saved-stuff''

\section{Mechanized security proof}
\label{sec:mechanized}

In this section, we describe our machine-checked proof of HMAC-DRBG's pseudorandomness. All definitions and theorem statements are taken straight from the Coq proofs and may be looked up by name in the file. See the end of \S\ref{coq-proofs} for the link to our repository.

\subsection{Foundational Cryptography Framework}

Our proof uses the Foundational Cryptography Framework \cite{petcher15:post}, a Coq library for reasoning about the security of cryptographic schemes. FCF provides a probabilistic programming language for describing cryptographic constructions, security definitions, and assumed-hard problems. Probabilistic programs are described using Gallina, the purely functional programming language of Coq, extended with a computational monad that adds uniform sampling on bit vectors. In this language, \texttt{\{0,1\}{\textasciicircum}n} describes uniform sampling on bit vectors of length n, and monadic arrows (e.g. \texttt{<-\$}) sequence the operations.
\texttt{Comp A} is the type of probabilistic computations that return values of type \texttt{A}, and a denotational semantics relates every \texttt{Comp A} with a probability distribution on \texttt{A}. 

FCF provides a theory of probability distributions, a program logic, and a library of reusable arguments and tactics that can be used to complete proofs. All of these items are based on theory derived from the semantics of probabilistic programs, and the result is a versatile proof framework with a small foundation. For example, FCF includes a reusable hybrid argument on lists that can be used in proofs like one described in this paper. FCF also provides a standard type and associated theory for an adversary that is allowed to interact with an oracle (called \texttt{OracleComp}).

\subsection{Cryptographic specification}
\label{cryptospec}

At a high level, any DRBG in the NIST specification takes as input an initial state \kv{} and a list of natural numbers, where each number is the number of blocks (bitvectors) requested for that call, e.g. $[1,3,2]$. It returns a list of lists of bitvectors with the correct number of generated blocks in each list, e.g. queried with $[1,3,2]$, our model of the DRBG would return $[[b_1], [b_2, b_3, b_4], [b_5, b_6]]$, where each $b_i$'s size is the output size of the PRF used in the DRBG. Because we model a nonadaptive adversary, there is only this one round of query-response.  

\paragraph{DRBG algorithm definitions.}
We begin by defining the abstract PRG used in our mechanized proof of pseudorandomness. Let $f$ be an arbitrary PRF that takes a key bitvector, an input list, and returns an output bitvector. $\eta$ is the output size of the PRF, $\mathit{Bvector} \, \eta$ is the type of bitvectors of size $\eta$ (also called ``blocks''), and $KV$ is a type synonym for the PRG's internal state, which is a tuple consisting of two blocks \kv.
\begin{lstlisting}
Variable $\eta$ : nat.
Hypothesis nonzero_eta : $\eta$ > 0.
Variable f : Bvector $\eta$ -> list bool -> Bvector $\eta$.
Definition KV : Set := Bvector $\eta$ * Bvector $\eta$.
\end{lstlisting}

The function \emph{Instantiate} instantiates the PRG's internal state \kv{} with ideal randomness. This idealized \instantiate{} function does not reflect the NIST specification; we discuss this in Section~\ref{thebad}.

\begin{lstlisting}
Definition Instantiate : Comp KV :=
  k <-$\$ $ {0,1}^$\eta$;
  v <-$\$ $ {0,1}^$\eta$;
  ret (k, v).
\end{lstlisting}

The function \emph{Gen\_loop} corresponds to the main loop of HMAC-DRBG's \emph{Generate} function. Given an internal state \kv{} and a number of requested blocks $n$, it returns a list of generated blocks, along with the last generated block, which becomes the new $v$ in the internal state.
\pagebreak
\begin{lstlisting}
Fixpoint Gen_loop (k : Bvector $\eta$) (v : Bvector $\eta$)
 $~~~~~~~~~$ (n : nat) : list (Bvector $\eta$) * Bvector $\eta$ :=
  match n with
  | O => (nil, v)
  | S n' => let v' := f k (to_list v) in
             let (bits, v'') := Gen_loop k v' n' in
             (v' :: bits, v'')
  end.
\end{lstlisting}

NIST specifies that at the end of a call to \emph{Generate}, \emph{Update} is immediately called to refresh the internal state. In our Coq definition, we have inlined \emph{Update} into \emph{Generate}:

\begin{lstlisting}
Definition Generate (state : KV) (n : nat) :
  $\qquad \qquad \qquad \qquad $ Comp (list (Bvector $\eta$) * KV) :=
  [k, v] <-2 state;
  [bits, v'] <-2 Gen_loop k v n;
  k' <- f k (to_list v' ++ zeroes);
  v'' <- f k' (to_list v');
  ret (bits, (k', v'')).
\end{lstlisting}
The NIST spec allows $n$ to be any number of bits (up to a specified maximum),
by discarding extra generated bits up to a multiple of the HMAC
block size.  Our security theorem measures $n$ in blocks, not bits;
it would be straightforward to extend the theorem to account
for discarding surplus bits.

\noindent \emph{Generate} and \emph{Instantiate} comprise our definition of the PRG.

The adversary $A$ is given the output of the PRG, which is a list of lists of blocks. It may perform a probabilistic computation on that input, and it returns a boolean guess attempting to distinguish the PRG's output from ideally random output. $A$ is nonadaptive; it is given the output all at once and cannot further query the PRG.

\begin{lstlisting}
Variable A : list (list (Bvector $\eta$)) -> Comp bool.
\end{lstlisting}

\noindent\textbf{Assumptions.}
We assume that there is a nonzero number of calls \emph{numCalls} to \emph{Generate}. We assume that on each call, the adversary requests the same nonzero number of blocks \emph{blocksPerCall}. This request list is called \emph{requestList}.

\begin{lstlisting}
Variable blocksPerCall : nat.
Variable numCalls : nat.
Definition requestList : list nat :=
  $\qquad \qquad\qquad $ replicate numCalls blocksPerCall.
Hypothesis H_numCalls : numCalls > 0.
Hypothesis H_blocksPerCall : blocksPerCall > 0.
\end{lstlisting}

Using the definitions of the PRG's functions and of the adversary, we define the real-world game $G_\mathrm{real}$ that simulates use of the PRG. First, the internal state of the PRG is instantiated. Then the \emph{Generate} function is iterated over the request list using the function \emph{oracleMap}, which tracks the PRG's internal state $(k, v)$ and updates it after each iteration. The resulting list of lists of blocks is passed to the adversary, which returns a boolean guess.

\begin{lstlisting}
Definition G_real : Comp bool :=
  [k, v] <-$\dol$2 Instantiate;
  [bits, _] <-$\dol$2 oracleMap Generate (k, v) requestList;
  A bits.
\end{lstlisting}

The ideal-world game $G_\mathrm{ideal}$ simply generates a list of lists of ideally random blocks, where each sublist is the same length as the corresponding one in the output generated by the PRG.  Here, \li|compMap| maps a
probabilistic function over a list.

\begin{lstlisting}
Definition G_ideal : Comp bool :=
  bits <-$\$ $ compMap Generate_rb requestList;
  A bits.
\end{lstlisting}

\subsection{Formal statement of theorem} \label{sec:statement}

Our main theorem states that the probability that an adversary can distinguish between
the real-world game $G_\mathrm{real}$ and
the ideal-world game $G_\mathrm{ideal}$ is bounded by a negligible quantity. 

\begin{lstlisting}
Theorem G_real_ideal_close :
  | Pr[G_real] - Pr[G_ideal] | $\le$ 
       numCalls * (PRF_Advantage_Max + Pr_collisions).
\end{lstlisting}

This quantity is the number of calls $numCalls$ to the PRG multiplied by a constant upper bound on the difference between adjacent hybrids, which is the sum of (1) an upper bound on the PRF advantage for HMAC, and (2) an upper bound on the probability of collisions in a list of length $blocksPerCall$ of $\mathit{Bvector} \s \eta$ that are each sampled uniformly at random.

For quantity (1), the PRF advantage of a function $f$ is informally defined to be the maximum probability that any p.p.t. adversary can distinguish $f$ from a PRF. Since a PRF's advantage is defined over all p.p.t. PRF adversaries, we construct a PRF adversary that takes an oracle as input (denoted by \emph{OracleComp} in its type), runs the PRG using that oracle, and returns the PRG adversary $A$'s guess. Technically, \emph{PRF\_Adversary} defines a family of adversaries parametrized by $i$, the index of the hybrid. There is actually a different PRF advantage for each hybrid $i$.

\begin{lstlisting}
Definition PRF_Adversary (i : nat) :
 $\hspace{0.8in}$ OracleComp Blist (Bvector eta) bool :=
  bits <--$\dol$ oracleCompMap_outer (Oi_oc' i) requestList;
  $\dol$ A bits.
\end{lstlisting}

The rationale behind this definition of \emph{PRF\_Adversary} will be given in the discussion of lemma \emph{Gi\_prf\_rf\_close} in Section~\ref{sec:fullproof}.

Now we can use the definition of \emph{PRF\_Adversary}, along with the PRF $f$, to define $\mathit{PRF\_Advantage\_Game} \s i$, which is the PRF advantage for hybrid $i$ (again discussed further in \emph{Gi\_prf\_rf\_close}).

\begin{lstlisting}
Definition PRF_Advantage_Game i : Rat := 
  PRF_Advantage RndK ({0,1}^$\eta$) f (PRF_Adversary i).
\end{lstlisting}

Finally, we define the upper bound on $f$'s PRF advantage to be
the maximum (taken over $i$) of the PRF advantages of hybrid $i$.

\begin{lstlisting}
Definition PRF_Advantage_Max := 
 PRF_Advantage_Game (argMax PRF_Advantage_Game numCalls).
\end{lstlisting}

If $f$ is instantiated with HMAC, then $PRF\_Advantage\_Max$ can be instantiated with the bound on the PRF advantage of HMAC derived in~\cite{beringer15:hmac}, modulo some technical details of instantiating the $PRF\_Adversary$ between the two definitions. 

Quantity (2) in the theorem statement is the upper bound on the probability of collisions in a list of length \emph{blocksPerCall} of $\mathit{Bvector} \s \eta$ that are each sampled uniformly at random.

\begin{lstlisting}
Definition Pr_collisions := (1 + blocksPerCall)^2 / 2^$\eta$.
\end{lstlisting}
This arises from bounding the probability of the ``bad event'' in intermediate hybrids, the probability of collisions in the inputs to a random function as used in that intermediate hybrid.

\subsection{Proof that HMAC-DRBG is pseudorandom} \label{sec:fullproof}

We formulate the real-world and ideal-world outputs as games involving interactions with the adversary $A$, and use the code-based game-playing technique \cite{bellare2006security}
to bound the difference in probability that the adversary returns \emph{True}. We call this ``bounding the distance between the two games.'' We use ``equivalent'' to mean that two games correspond to the same probability distribution.  We abbreviate ``random function'' as ``RF'' and ``ideal randomness'' as ``RB,'' for ``random bits.'' The output block-size of the PRF is denoted by $\eta$.

Now we outline the proof of the main result $\emph{G\_real\_ideal\_close}$. Here is the
proof script:

\begin{lstlisting}
Theorem G_real_ideal_close :
  | Pr[G_real] - Pr[G_ideal] | $\le$ 
       numCalls * (PRF_Advantage_Max + Pr_collisions).
Proof.
  rewrite Generate_move_v_update.
  rewrite G_real_is_first_hybrid.
  rewrite G_ideal_is_last_hybrid.
  specialize (distance_le_prod_f (fun i => Pr[Gi_prg i]) Gi_adjacent_hybrids_close numCalls).
  intuition.
Qed.
\end{lstlisting}

\kye{I suggest we rename $G1\_prg$ to $G\_real\_v$, and rename $Gi\_prg$ to $hybrid$, in both the paper and the Coq file. I don't have time to do it but I think it would greatly help readability.}

The proof is quite short, since it consists of applying five lemmas: rewriting $G_\mathrm{real}$ into the nicer form defined by $\emph{G1\_prg}$, then rewriting the ideal-world and real-world games as hybrids, then applying the theorem \emph{Gi\_adjacent\_hybrids\_close} that bounds the difference betweeen any two adjacent hybrids by \emph{Gi\_Gi\_plus\_1\_bound}. Most of the work is done in \emph{Gi\_adjacent\_hybrids\_close}.

Our proof is structured according to this outline,
and we will present the lemmas in this order:

% \begin{figure}
      \begin{longenum}
\item[] \hspace{-2em}$*$~\li|G_real_ideal_close|
      \item \li|Generate_move_v_update|
      \item \li|G_real_is_first_hybrid|
      \item \li|G_ideal_is_last_hybrid|

      \item \li|Gi_adjacent_hybrids_close|
        \begin{longenum}
        \item \li|Gi_prog_equiv_prf_oracle|
        \item \li|Gi_replace_prf_with_rf_oracle|
        \item \li|Gi_replace_rf_with_rb_oracle|
          \begin{longenum}
          \item \li|Gi_prog_equiv_rb_oracle|
          \item \li|Gi_rb_rf_identical_until_bad|
            \begin{longenum}
            \item \li|fcf_fundamental_lemma|
            \item \li|Gi_rb_rf_return_bad_same|
            \item \li|Gi_rb_rf_no_bad_same|
            \end{longenum}
          \item \li|Gi_Pr_bad_event_collisions|
          \end{longenum}
        \end{longenum} 

      \item \li|hybrid_argument| 
      \end{longenum}
% \end{figure}

We start by proving a lemma to rewrite $G_\mathrm{real}$ into a nicer form defined by \emph{G1\_prg}.

\begin{lstlisting}
Lemma Generate_move_v_update :
  Pr[G_real] == Pr[G1_prg].
\end{lstlisting}

The way NIST's \emph{Generate} function updates $v$ makes it difficult to write neat hybrids. Specifically, the PRF is re-keyed on the second line of \emph{Generate}; then, on the next line, in order to update $v$, \emph{Generate} queries the \emph{new} PRF. Therefore, if we were to replace $f_k$ with a PRF oracle, it would have to span the last line of \emph{Generate} and the first two lines of the next \emph{Generate} call, but not include the second call's $v$-update. This would be messy to reason about. We solve this problem by moving each $v$-update to the beginning of the next call of \emph{Generate}, then prove that the outputs are still identically distributed. Then, after the PRF is re-keyed, we do \emph{not} further query it in this call of \emph{Generate\_v}. 
% \awa{Harmonize with the previous discussion of this in section 3?}

In the intermediate game \emph{G1\_prg}, we move each terminal $v$-update to the beginning of the next call to \emph{Generate} by splitting the \emph{Generate} function into \emph{Generate\_noV} and \emph{Generate\_v}. 

The function \emph{Generate\_noV} is used for the first call. It differs from \emph{Generate} because it does not start by updating the $v$. 

\begin{lstlisting}
Definition Generate_noV (state : KV) (n : nat) :
  $\qquad \qquad \qquad $   Comp (list (Bvector $\eta$) * KV) :=
  [k, v] <-2 state;
  [bits, v'] <-2 Gen_loop k v n;
  k' <- f k (to_list v' ++ zeroes);
  ret (bits, (k', v')).
\end{lstlisting}

The second version, \emph{Generate\_v}, starts by updating $v$ and does not update $v$ again. Then, as in $G_\mathrm{real}$, the resulting list of lists of blocks is passed to the adversary, which returns a boolean guess. 

\begin{lstlisting}
Definition Generate_v (state : KV) (n : nat) :
  $\qquad \qquad $   Comp (list (Bvector $\eta$) * KV) :=
  [k, v] <-2 state;
  v' <- f k (to_list v);        (* new *)
  [bits, v''] <-2 Gen_loop k v' n;
  k' <- f k (to_list v'' ++ zeroes);
  ret (bits, (k', v'')).
\end{lstlisting}

This is the revised real-world game that uses the two split-up \emph{Generate} versions defined above. 

\begin{lstlisting}
Definition G1_prg : Comp bool :=
  [k, v] <-$\$ $2 Instantiate;
  [head_bits, state'] <-$\$ $2 Generate_noV (k, v)
                                blocksPerCall;
  [tail_bits, _] <-$\$ $2 oracleMap Generate_v state' 
                                (tail requestList);
  A (head_bits :: tail_bits).
\end{lstlisting}

To prove the real-world game $G_\mathrm{real}$ equivalent to the new version \emph{G1\_prg}, we prove that the pseudorandom output produced by this sequence of calls
$$[\mathit{Generate}, \mathit{Generate}, \ldots]$$
is identically distributed to the pseudorandom output produced by this new sequence of calls
$$[\mathit{Generate\_noV}, \mathit{Generate\_v}, \mathit{Generate\_v}, \ldots].$$
% (The intermediate $(k,v)$ state between the PRG calls is \emph{not} the same, but the two outputs are still identically distributed.)

Using $\mathit{G1\_prg}$, to set up the hybrid argument, we rewrite the two ``endpoint'' games as hybrids. We first prove (by a straightforward induction) that the real-world game is equal to the first hybrid. 

\begin{lstlisting}
Lemma G_real_is_first_hybrid :
  Pr[G1_prg] == Pr[Gi_prg O].
\end{lstlisting}

Next: the ideal-world game is equal to the last hybrid,
also proved by a straightforward induction.

\begin{lstlisting}
Lemma G_ideal_is_last_hybrid :
  Pr[G_ideal] == Pr[Gi_prg numCalls].
\end{lstlisting}

%%%%%%%%%%%%%%%%%%%% needs rewrite

% This should really be a triple-subsection
\noindent\textbf{Bounding the difference between adjacent hybrids.}
As described in Section~\ref{sec:statement}, the difference between any two adjacent hybrids is defined to be the sum of the maximum PRF advantage over the PRF advantage for each hybrid, plus an upper bound on the probability of collisions in a list of length \emph{blocksPerCall} of uniformly randomly sampled blocks of length $\eta$ each.

\begin{lstlisting}
Theorem Gi_adjacent_hybrids_close : forall (n : nat),
  | Pr[Gi_prg n] - Pr[Gi_prg (1+n)] |
   $\le$ PRF_Advantage_Max + Pr_collisions.
\end{lstlisting}

% TODO picture, or link to existing picture
The game $\mathit{Gi\_prg} \s i$ defines the $i$th hybrid game, which replaces the output of the PRF with ideal randomness for any call before the $i$th call.
% That is, if $numCalls = 3$ and ...

\begin{lstlisting}
Definition Gi_prg (i : nat) : Comp bool :=
  [k, v] <-$\$ $2 Instantiate;
  [bits, _] <-$\$ $2 oracleMap (choose_Generate i) 
                       (O, (k, v)) requestList;
  A bits.
\end{lstlisting}
To do this, $\mathit{Gi\_prg} \s i$ uses an oracle $\mathit{choose\_Generate} \s i$ that tracks the number of calls \emph{callsSoFar} (in addition to the PRG's internal state \kv{}) and returns the appropriate version of the PRG for that call. $\mathit{choose\_Generate} \s i$ returns \emph{Generate\_rb\_intermediate} if the call index is less than $i$, and \emph{Generate\_noV} or \emph{Generate\_v} as appropriate otherwise.

\begin{lstlisting}
Definition choose_Generate (i : nat) (sn: nat * KV)
  (n: nat) : Comp (list (Bvector $\eta$) * (nat * KV)) :=
  [callsSoFar, state] <-2 sn;
  let Gen := if lt_dec callsSoFar i
             then Generate_rb_intermediate
             else if beq_nat callsSoFar O then Generate_noV
             else Generate_v in
  [bits, state'] <-$\$ $2 Gen state n;
  ret (bits, (1+callsSoFar, state')).
\end{lstlisting}

The two quantities that compose the bound between any two adjacent hybrids $\mathit{Gi\_prg} \s i$ and $\mathit{Gi\_prg} \s (i+1)$ are defined as follows (see also Section~\ref{sec:statement}).

\begin{lstlisting}
Definition PRF_Advantage_Max := 
 PRF_Advantage_Game (argMax PRF_Advantage_Game numCalls).
Definition Pr_collisions := (1 + blocksPerCall)^2 / 2^$\eta$.
\end{lstlisting}

To bound the distance between adjacent hybrids $Gi\_prg \s i$ and $Gi\_prg \s (i+1)$, we derive the two quantities in the bound in two steps. 
\emph{PRF\_Advantage\_Max} is derived in lemma
\emph{Gi\_prf\_rf\_close}, and 
\emph{Pr\_collisions} in lemma \emph{Gi\_rf\_rb\_close}.

% \subsubsection*{Tree lemma 4a} \hfill\\ 

Before we can derive the bound, we first need to transform our normal hybrid definition $\mathit{Gi\_prg} \s i$ into the equivalent oracle-using definition
$\mathit{Gi\_prf} \s i$.

\begin{lstlisting}
Lemma Gi_prog_equiv_prf_oracle : forall (i : nat),
  Pr[Gi_prg i] == Pr[Gi_prf i].
\end{lstlisting}

The only difference between the two definitions is that $\mathit{Gi\_prf} \s i$ is written in the form of an oracle computation that uses a provided oracle on the $i$th call. In the new game, the oracle we provide is the PRF, defined in the code as $\mathit{f\_oracle} \s f \s k$. The old game $Gi\_prg \s i$ already uses the PRF on the $i$th call by definition.

\begin{lstlisting}
Definition Gi_prf (i : nat) : Comp bool :=
  k <-$\dol$ RndK;
  [b, _] <-$\dol$2 PRF_Adversary i (f_oracle f k) tt;
  ret b.
\end{lstlisting}

The computation \emph{PRF\_Adversary} is a constructed adversary against the PRF that rewrites our PRG execution, previously defined using \emph{oracleMap}, in terms of \emph{oracleCompMap}, which is a version of \emph{oracleMap} that replaces the PRF on the $i$th call of hybrid $i$ with a provided oracle. The provided oracle is the PRF, $\mathit{f\_oracle} f k$.

\begin{lstlisting}
Definition PRF_Adversary (i : nat) 
  : OracleComp Blist (Bvector eta) bool :=
  bits <--$\dol$ oracleCompMap_outer (Oi_oc' i) requestList;
  $\dol$ A bits.
\end{lstlisting}

The technical lemma \emph{Gi\_normal\_prf\_eq}
allows us to provide different oracles to the new oracle-taking form of the hybrid to use on its $i$th call, which we use to define the two intermediate hybrids after $\mathit{Gi\_prg} \s i$.

% \subsubsection*{Tree lemma 4b} \hfill\\ 

We derive the first quantity used in the bound,\\ 
\emph{PRF\_Advantage\_Max}.

\begin{lstlisting}
Lemma Gi_replace_prf_with_rf_oracle : forall (i : nat),
    i $\le$ numCalls ->
| Pr[Gi_prf i] - Pr[Gi_rf i] | $\le$ PRF_Advantage_Max.
\end{lstlisting}

Recall that we have just rewritten the $i$th hybrid $\mathit{Gi\_prg} \s i$ as \emph{Gi\_prf}, defined above, which uses the PRF oracle $\mathit{f\_oracle} \s f \s k$. In \emph{Gi\_rf}, we replace the PRF oracle with a random function oracle \emph{randomFunc}.

\begin{lstlisting}
Definition Gi_rf (i : nat) : Comp bool :=
  [b, _] <-$\dol$2 PRF_Adversary i (randomFunc ({0,1}^eta)) nil;
  ret b.
\end{lstlisting}

A random function is simply a function that, queried with any new input, returns an ideally random $\mathit{Bvector} \s \eta$ and caches the result. If it is queried with a previous input, it returns the previously-sampled bit-vector.
By the definition of PRF advantage, for any adversary against the PRF, the probability that the adversary can distinguish between the PRF and a random function is \emph{PRF\_Advan\-tage\_\linebreak[2]Game}$(i)$, which is upper-bounded by the maximum taken over $i$ (over all hybrids).

Having bounded the distance between $\mathit{Gi\_prg} \s i$ and the intermediate hybrid $\mathit{Gi\_rf} \s i$, we proceed to derive the second quantity in the bound: the distance between $\mathit{Gi\_rf} \s i$ and the next normal hybrid $\mathit{Gi\_prg} \s (i+1)$ is \emph{Pr\_collisions}.

\begin{lstlisting}
Lemma Gi_replace_rf_with_rb_oracle : forall (i : nat), 
  | Pr[Gi_rf i] - Pr[Gi_prg (1+i)] | $\le$ Pr_collisions. 
\end{lstlisting}

Similarly to lemma \emph{Gi\_prf\_rf\_close}, we must first rewrite the leftmost hybrid $\mathit{Gi\_prg} \s (i+1)$ in the form of the oracle-using hybrid \emph{Gi\_rb} that uses an oracle \emph{rb\_oracle} on its $i$th call. 

\begin{lstlisting}
Lemma Gi_prog_equiv_rb_oracle : forall (i : nat),
    Pr[Gi_prg (1+i)] == Pr[Gi_rb i].
\end{lstlisting}

\begin{lstlisting}
Definition Gi_rb (i : nat) : Comp bool :=
  [b, state] <-$\dol$2 PRF_Adversary i rb_oracle nil;
  let rbInputs := fst (split state) in  ret b.
\end{lstlisting}

This new oracle \emph{rb\_oracle} returns an ideally random $\mathit{Bvector} \s \eta$ for every query, not caching any of its queried inputs. 

\begin{lstlisting}
 Definition rb_oracle (state : list (Blist * Bvector $\eta$))
 $\hspace{2in}$  (input : Blist) :=
  output <-$\dol$ ({0,1}^eta);
  ret (output, (input, output) :: state).
\end{lstlisting}

We prove this rewrite, \emph{Gi\_normal\_rb\_eq}, via another induction.
We use this rewrite to bound the difference between the two oracle-using hybrids $\mathit{Gi\_rf} \s i$ and $Gi\_rb \s i$. The two games differ in only one aspect: on the $i$th call, $\mathit{Gi\_rf} \s i$ replaces the PRF with a random function oracle, whereas $\mathit{Gi\_rb} \s i$ replaces the PRF with an ideally random oracle. These two oracles themselves can be distinguished by an adversary in only one way: if the random function is queried with an input that has been previously queried, then it becomes deterministic and returns the same output as before, whereas the ideally random oracle will still return ideally random output.

Therefore, the distance between $\mathit{Gi\_rf} \s i$ and $\mathit{Gi\_rb} \s i$ ought to be bounded by the probability that there are duplicates in the inputs to the random function, which we refer to as the probability of the \emph{bad event} occurring. The probability that the adversary can distinguish between $\mathit{Gi\_rf} \s i$ and $\mathit{Gi\_rb} \s i$ should then be upper-bounded by the probability of the bad event.

We formalize this intuitive argument by applying the \emph{fundamental lemma of game-playing}~\cite{bellare2006security}, which states that the difference between two games that are ``identical until bad'' is upper-bounded by the probability of the ``bad event'' happening. Two games are ``identical until bad'' if they are syntactically identical except for statements that follow the setting of a flag \emph{bad} to \emph{true}.

We state the bound on the difference between the game in a form following that of the fundamental lemma:

\begin{lstlisting}
Lemma Gi_rb_rf_identical_until_bad : forall (i : nat),
  | Pr[x <-$\dol$ Gi_rf_dups_bad i; ret fst x]
       - Pr[x <-$\dol$ Gi_rb_bad i; ret fst x] |
  $\le$ Pr[x <-$\dol$ Gi_rb_bad i; ret snd x].
\end{lstlisting}

The new games \emph{Gi\_rf\_dups\_bad} and \emph{Gi\_rb\_bad} are just versions of the original games rewritten to return a tuple, where the first element is the usual output of $\mathit{Gi\_rf}$ or $\mathit{Gi\_rb}$ and the second event is a boolean that indicates whether the bad event has occurred. To apply the fundamental lemma to prove \emph{Gi\_rb\_rf\_identical\_until\_bad}, we first prove that the games \emph{Gi\_rf\_dups\_bad} and \emph{Gi\_rb\_bad} are indeed identical until bad.
First we prove that the bad event has the same probability of happening in each game:

\begin{lstlisting}
Lemma Gi_rb_rf_return_bad_same : forall (i : nat),
    Pr  [x <-$\dol$ Gi_rb_bad i; ret snd x ] ==
    Pr  [x <-$\dol$ Gi_rf_dups_bad i; ret snd x ].
\end{lstlisting}
This is clearly true (but don't take our word for it: ask Coq!). Second, we prove that if the bad event does not happen, the two games have identically distributed outputs (evalDist takes a Comp and produces the corresponding probability mass function).

\begin{lstlisting}
Lemma Gi_rb_rf_no_bad_same : forall (i: nat) (a: bool), 
     evalDist (Gi_rb_bad i) (a, false)  ==
     evalDist (Gi_rf_dups_bad i) (a, false).
\end{lstlisting}
This, too, is clearly true. Thus, we can apply the fundamental lemma to derive the upper bound:\\
\indent \lstinline{Pr[x <-$\dol$ Gi_rb_bad i; ret snd x]}.

That concludes the identical-until-bad argument, which bounded the difference between the two games by the probability of the bad event. We finish the proof of the main result by calculating an upper bound on the probability of the bad event, which is, again, a collision in the inputs to the RF. 

\begin{lstlisting}
Theorem Gi_Pr_bad_event_collisions: 
   forall (i:nat) (v: Bvector $\eta$),
   Pr  [x <-$\dol$ Gi_rb_bad i; ret snd x ] $\le$ Pr_collisions.
\end{lstlisting}

% \subsubsection*{Tree lemma 5} \hfill\\ 

On the $i$th call to the PRG, the provided oracle (here the random function) is queried once with a constant $v$ and then \emph{blocksPerCall} times after that, where each RF input is actually a previous output from the RF due to the design of HMAC-DRBG's \emph{Gen\_loop}. If an RF output is not a duplicate, it is ideally random.  (To derive a more general bound, \emph{blocksPerCall} can be generalized to the $i$th element in the request list. The rest of the proof is unchanged.)

Therefore, the probability of collisions in the inputs to the RF simplifies to the probability of collisions in a list of length $\mathit{blocksPerCall} + 1$ of elements $\mathit{Bvector} \s \eta$ that are each sampled uniformly at random. By the union bound, this is less than ${(1+\mathit{blocksPerCall})^2}/{2^\eta}$.

We have bounded the distance between adjacent hybrids with \emph{Gi\_adjacent\_hybrids\_close}, so we conclude the whole proof by applying the hybrid argument (which is essentially the triangle inequality, repeatedly applied) with that lemma to bound the distance between the first and last hybrid.

% TODO I changed the theorem statement to make it more readable, but it's not reflected in FCF yet.
\begin{lstlisting}
Theorem hybrid_argument :
       forall (G : nat -> Rat) (bound : Rat),
       (forall i : nat, | G i - G (1+i) | $\le$ bound) ->
       forall n : nat, | G 0 - G n | $\le$ n * bound
\end{lstlisting}

This generic theorem states that if the distance between adjacent games $G \s i$ and $G \s (i+1)$ is bounded by the rational number $bound$, the difference between the first hybrid $G \s 0$ and the last hybrid $G \s n$ is bounded by $n \cdot \mathit{bound}$.

%%%%%%%%%%%%%%%%%%%%%%%%%%%%%%%%%%%%%%%%%%%%%%%%%%%%%%%%%%%%%%%%%%%%%%%%%%%%%%%%%%%%%

% TODO merge with another subsection?
\subsection{Coq-specific proofs}

Working in Coq requires doing many proofs of program and game equivalence. These program equivalence proofs are generally elided, by choice or by accident, in paper proofs. Two such proofs are \li|Gi_prog_equiv_prf_oracle| and
\li|Gi_prog_equiv_rb_oracle|. Because using oracles changes the types of all computations that use that computation, the types of \emph{Gen\_loop}, \emph{Generate\_v}, \emph{choose\_Generate}, and \emph{oracleMap} all change from \emph{Comp} to \emph{OracleComp}, and we write new functions with the new type and prove them each equivalent. Proving the two lemmas requires doing complex inductions over both the list of blocks provided by the adversary and over each number of blocks in a call. 

Semiformal crypto proofs (e.g., \cite{hirose2008security})
often leave the burden  of checking for semantic equivalence on the reader. 
An unintentional error in this part of a crypto proof is unlikely,
but an intentional error is possible, and it may be subtle enough that a standards body would not notice it. 

\subsection{Efficiency of adversary}

We do not formally prove the efficiency of the constructed adversaries in the proof. It is necessary to inspect the family of constructed adversaries (that is, Definition PRF\_Adversary in \S\ref{maintheorem})
in order to establish that it is probabilistic polynomial time;
which clearly it is
(assuming the adversary against the PRG is p.p.t.).
While this omission increases the amount of code that must be inspected, it also allows the result to be interpreted in any model of computation and cost.

\section{Concrete security estimate}
\label{concrete}
The distance between $G_\mathrm{real}$ and $G_\mathrm{ideal}$ is
\[\mathit{numCalls}\cdot(\mathit{PRF\_Advantage\_Max} + (1+\mathit{blocksPerCall})^2/2^\eta).\]  What does
that mean concretely about an adversary's ability to predict
the output of HMAC-DRBG?  Let us assume the adversary runs in time
and space at most $2^t$.

Beringer \emph{et al.} proved that the PRF advantage for HMAC is the
sum of the advantages associated with these three assumptions:
(1) SHA-256 is weakly collision-resistant, (2) the SHA-256 compression function
is a PRF, and (3) the dual family of the SHA256 compression function is a PRF against a (very) restricted related-key attack. 

(1) If the best method for finding collisions is a birthday attack,
a collision will be found with probability at most $2^{2t-256}$.
\ (2) If the best method for distinguishing SHA-256's compression from
a random function is an exhaustive search on keys,
the probability of success is at most $2^{t-256}$.
\ (3) If, in addition, we consider that the constructed adversary in the proof gains little or no additional advantage from the related key oracle, then the probability of success is at most $2^{t - 256}$ \cite{bellare2006new}.

From these assumptions, we can conclude that the
adversary's advantage in distinguishing HMAC-SHA256 from a random function is
at most $2^{2t - 256} + 2^{t - 255}$. For $t = 78$,
the advantage is at most $2^{-100} + 2^{-177}$.

Assuming there are $2^{48}$ calls to the DRBG \kye{NIST specifies a maximum number of calls to the DRBG and maximum number of blocks per call. See the section on the concrete bound in my thesis, or look at the NIST spec. You may want to use those numbers. You might also cite Hirose's concrete bound.}
with block length $\eta=128$ and $\mathit{blocksPerCall}=10$,
then the distance is about $2^{48}\cdot(2^{-100}+2^{-177}+(1+10)^2/2^{128}$,
or about $2^{-52}$.  That is, the adversary with a thousand million terabytes ($< 2^{78}$ bits), and a million 1-gigahertz processors running for a year
($< 2^{78}$ cycles) has a $2^{-52}$ chance of guessing better than random.

It also means that block length $128$ is better than $64$,
but increases beyond that are not particularly useful.

Warning: These calculations are based on unproved assumptions (1)(2)(3)
about SHA-256 and the best ways to crack it.  The assumptions may not be
true, and there may be better ways to crack SHA-256.

\section{mbedTLS implementation is correct}
\label{sec:mbedtlsVerification}
\input{mbedtlsVerification}

\section{Lessons learned}
\label{lessons}
\label{goodbadugly}
Three NIST design decisions made the proof harder or easier.\\

% TODO add code and cite lemma name
\noindent\textbf{The good.} NIST re-keys the PRF with a length-extended input. This simplifies the component of the proof that bounds the probability of collisions in inputs to the random function. Because HMAC can take inputs of any length, and all previous inputs for HMAC with that key were of fixed length (since HMAC has a fixed output), it is easy to prove that the length-extended input cannot be a duplicate of any previous input to the random function.\\

\noindent\textbf{The bad.}
\label{thebad}
In \emph{Instantiate}, NIST begins the initialization of the state \kv{} as follows:\\
\begin{game}
\item[] $k \leftarrow \mathit{HMAC}_{c_1}(c_2 || \mathtt{0x00} || \mathit{entropy\_input})$ 
\item[] $v \leftarrow \mathit{HMAC}_{k}(c_2)$\\
\end{game} 
If we assume that \textit{entropy\_input} is ideally random,
then it would be better to incorporate it into the first
argument ($c_1$) of HMAC, rather than the second,
because the latter forces us to assume or prove that HMAC
is an entropy extractor.

If \textit{entropy\_input} is \emph{not} ideally random (as would likely be the case in a real system),
then we will indeed need an entropy extractor.
HMAC seems to be as good a candidate as any.
However, existing proofs of HMAC as a randomness extractor
\cite{dodis2004randomness,fouque2008hmac,krawczyk2010cryptographic}
are only for keys $c_1$ that are
random, so they may not apply here because HMAC-DRBG's $c_1=0$.
So we have no security proof of
\emph{Instantiate}.  Instead, we assume that the $(k,v)$ produced
by \emph{Instantiate} is ideally random.\\

% \footnote{We have a VST proof that that the
%   mbedTLS implementation of HKDF is functionally correct with respect
%   to IETF RFC 5869, but that is beyond the scope of this paper.}

% TODO add code and confirm with matt
\noindent\textbf{The ugly.}
It is hard to reason about NIST's decision to update the PRG's internal state vector $v$ immediately after re-keying the PRF in the same call to \emph{Generate}. This decision makes it difficult to define hybrids that replace a PRF on each call, since the PRF changes and is used before a call is over. We solve this problem by moving each $v$-update to the beginning of the next call and proving that the new sequence of programs is indistinguishable to the adversary, and we continue the proof with the rewritten \emph{Generate}.\\

% ???
% This is because at the beginning of each call to \emph{Generate}, by the hybrid argument, we can assume that the \kv{} have been randomly sampled. Now we have to replace the key with a random key to reason about updating the $v$.

\noindent\textbf{How verification helped.} The process of mechanizing our proof and working in Coq and FCF enabled us to state definitions, functions, and games in a language that was expressive, precise, and executable. This benefited our proof in two main ways.

First, working in this environment prevented us from making subtle mistakes in dealing with a tricky hybrid argument. For example, for every lemma bounding the difference between adjacent hybrids $H_i$ and $H_{i+1}$, Coq requires the lemma to be parametrized by $i$. This helped early in the proof; it forced us to realize that there were multiple PRF adversaries: one for each hybrid $H_i$.
% , so it is also parametrized by $i$).
% Typical crypto proofs don't have multiple PRF adversaries, so we hadn't even considered this possibility.

Additionally, Coq prevented us from making suble off-by-one errors in the hybrid argument. It's easy to mix up the total number of hybrids for a list of blocks of $n$ elements---is it $n$ or $n+1$ hybrids? It's also easy to mix up whether the oracle is replaced in the $i$th call or the $i+1$th call. If the numbering is wrong, the relevant lemmas will become unprovable.

Second, working in this environment allowed us to ``execute'' computations in lemmas and test if they held true on small examples, which helped us spot typos in the proof. For example, in \emph{Generate\_v}, we had typed $\mathrm{ret} \s (\mathit{bits}, (k', v'))$ instead of $\mathrm{ret} \s (\mathit{bits}, (k', v''))$, which made one lemma unprovable until we fixed it. Similar bugs have been found in real-world code because a typo or mistake rendered the code unverifiable. \kye{Cite or delete. I'm 100\% sure I read an email on Coq-club about verifying some TLS-like thing that failed because of a typo in the code, but I can't dig it up.}\\

% These bugs have also been turned into real-world exploits.
% One example of this can be found in the lemmas whose names begin with \li|Gi_prog_equiv_prf_oracle_compspec|, where wrong results after computing on small expressions allowed us to unearth several typoes in the indices in $Oi\_oc'$ and $Generate\_v$.

\noindent\textbf{The trusted base.}
As explained in \S\ref{sec:intro}, a key part of
the trusted base of any machine-checked proof
is \emph{the statement of the theorem and assumptions.}
Our theorem is stated as 200 lines of Coq,
(pseudorandom properties of HMAC-DRBG),
on top of about 50 lines of supporting FCF definitions
defining the notions of pseudorandomness, etc.
We rely on assumptions related to collision-resistance of SHA \cite[\S 4.2]{beringer15:hmac}, which means that
the 169-line definition of SHA in Coq
is also part of the trusted base.

On the other hand, the definitions of the HMAC function, the HMAC-DRBG function,  and the specifications of internal APIs need not be trusted, because we prove an end-to-end theorem that, regardless of what these functions actually are,
they are correctly implemented by the C code  and they correctly provide a PRF.

\vspace{\baselineskip}
\noindent\textbf{Proof effort.}
The proof of cryptographic security spans about 4500 lines of Coq code---which need not be trusted because it is checked by machine,
but still it has a pleasantly clear structure as explained in \S\ref{sec:fullproof}.
Developing and mechanizing this proof took about
% Katherine: 4.5 person-months
% Naphat: 1 person-month
% Lennart: 1 person-month?
% Matt: 0.25
% Adam: 0.25
% TOTAL: 7
5 person-months.
The C program-correctness proof 
is about 10,200 lines of Coq, and
took 1 person-month
by a new user of Verifiable C.
Theorem~\ref{spec-equiv} (about 700 lines),
along with the connection to the HMAC proofs and other
configuration engineering (included in the 10,200 lines) took
about 1 person-month.
Coq checks all the proofs in about 16.5 minutes (running Coq 8.6 on a 2.9 GHz Intel Core i7 on a single core;  faster on multicore).
% That is:
% 2270 lines of AST+compspecs generated by CompCert, not include in these totals;
% 1200 lines of API spec, all but 25 lines of which is internal, thus part of the proof
% 715 lines of bridge proof (spec-equiv)
% 9015 lines of Verifiable C proofs

% Our crypto-security theorem is quite independent of the details
% of the C program, or even of its ``functional model.''  Therefore,
% our security proof can serve, in a modular way, for other implementations
% of HMAC-DRBG in other programming languages.

\section{Future work}
\label{future}

%We proved that the output of the HMAC-DRBG $\allowbreak \generate$ algorithm is pseudorandom given a uniformly random internal state, and we linked this proof to the security of a widely-used TLS library. 

In this work we focused primarily on the pseudorandomness property of the DRBG. In SP 800-90A, NIST proposes a variety of additional properties, including ``backtracking resistance'' and ``prediction resistance.'' The former property ensures that a compromise of generator state should not affect the security of output generated {\em prior} to the compromise, while the latter property addresses post-compromise security (under the assumption of a continually reseeded generator). Extending our proof to backtracking resistance should
be straightforward.  A proof that HMAC-DRBG is prediction-resistant
would require a careful specification of
the rate at which new entropy must be provided.

Our security proof leaves out some optional parameters such as the personalization nonce and additional input.  We hope to address these in the future.

While we have addressed only a single NIST DRBG construction, future work may prove the security of other widely-used SP 800-90A constructions such as CTR-DRBG and HASH-DRBG (based on block ciphers and hash functions respectively). These generators have similar constructions to HMAC-DRBG, and thus it is likely that our work may be adapted with reasonable modifications. 

HMAC-DRBG's Instantiate function
uses HMAC with fixed nonrandom key as a randomness extractor;
this is not known to be sound (see \S\ref{thebad}).

Concerning implementation correctness, we would like to extend our
framework by integrating provable protection against side-channel
attacks~\cite{DBLP:conf/ccs/BartheBCLP14,DBLP:conf/uss/AlmeidaBBDE16},
including protection against timing attacks that recently affected
implementations of MEE-CBC~\cite{DBLP:journals/iacr/AlmeidaBBD15}.

We hope this work will lead to formal verification efforts on other widely used cryptographic software packages. Ideally, this type of verification could replace or supplement the use of test vectors, which are currently used by organizations such as the U.S. Cryptographic Module Validation Program to validate DRBG correctness. 

%(Relation to forward security and backward security?)

%Prove HMAC is an entropy extractor and generalize Instantiate.

%Use the previous bound on PRF advantage.

%Prove HMAC-DRBG possesses prediction resistance.

%Verify DRBGs with different constructions.

%Build the tools to allow people to easily verify more cryptographic constructions and systems (DeepSpec?).

%Prove correctness and security WRT full NIST spec (parameters like provided data, nonce, other inputs).

\section{Conclusion}

We have created the first end-to-end formal security-and-correctness
verification of a real-world
random number generator,
and the first machine-checked proofs of HMAC-DRBG security and of
a correct implementation of HMAC-DRBG.  We can say with very high confidence
that HMAC-DRBG and its mbedTLS implementation are secure in their basic
functions against a nonadaptive adversary.
There are no bugs (including buffer overruns) and no backdoors,
purposeful or otherwise.
We have not proved security of the \emph{Instantiate} algorithm
(though we have proved that mbedTLS implements it correctly).
Our proof is modular; thus, it is portable to
other implementations of HMAC-DRBG.  Our proof should be extensible to other
features and properties of HMAC-DRBG, with limitations as discussed
in \S\ref{future}.

Machine-checked proofs of cryptographic security should be a required
part of the NIST standardization process.
This would likely have detected serious generator flaws such as those that occurred in Juniper NetScreen devices~\cite{Checkoway16}.  And though HMAC-DRBG is a good DRBG, proving the security of the construction at
NIST standardization time would likely have resulted in a more elegant construction, and would have allowed NIST 
to clearly specify the security assumptions on HMAC required by $\instantiate$.

Finally, implementations of cryptographic standards should come with
machine-checked proofs of correctness.  This would have prevented
the HeartBleed bug, and probably
prevented the Debian/OpenSSL fiasco (see the appendix).

% TODO move references to before appendix

% \section{Coq definitions}

% TODO

% \input{proofs}

%%%%%%%%%%%%%%%

% \begin{acks}

% TODO

% \end{acks}

%%% Local Variables:
%%% mode: latex
%%% TeX-master: "header"
%%% End:

%% file: mbedtlsVerification.tex
\subsection{Functional specification program} 
We verify the mbedTLS implementation of HMAC-DRBG with respect to a
functional program in Coq, with functions
\textbf{mbedtls\_\linebreak[2]generate}, \ldots,
\textbf{mbedtls\_\linebreak[2]reseed}.  Where the Coq functions
\textbf{Generate}, etc.  described in Section~\ref{cryptospec} operate
over $\mathit{Bvector} \, \eta$, these functions operate over finite
lists of sequences of integers (type \lstinline{list Z} in Coq), and
over abstract DRBG states of type

% \pagebreak

\begin{lstlisting}
Inductive hmac256drbgabs :=
  HMAC256DRBGabs: forall (key V: list Z)
       (reseedCounter entropyLen: Z)
       (predictionResistance: bool)
       (reseedInterval: Z), hmac256drbgabs.
\end{lstlisting} 
The functions' definitions are structured in accordance with NIST
800-90A, specialising generic DRBG functions to the functional
program for HMAC-SHA256 from Beringer \emph{et al.}~\cite{beringer15:hmac}.
The detailed definitions are available in our source code.

There is a straightforward relation $\sim$ between states
$(k,v)$ from Section~\ref{sec:mechanized}
and states \li|(HMAC256DRBGabs $k$ $v$ 0 0 false $2^{48}$)|
of our \textbf{mbedtls} functional specification.

Based on $\sim$, we formulate an equivalence theorem between
the two specifications of the \emph{Generate} function,
under certain conditions.  For example,
mbedTLS measures \texttt{out\_len} in bytes,  up to 1024,
while our crypto spec measures $n$ in 256-bit blocks;
the reseed counter must be less than the reseed interval;
the \emph{predictionResistance} flag must be false, and so on.

\begin{theorem}The functions
\label{spec-equiv}
\textbf{mbedtls\_generate}, etc. are equivalent
to \textbf{Generate}, etc.,
for $\sim$ similar states.
\end{theorem}
% ~\\
% \vspace{-10ex}
\begin{proof}By induction on iterations of the \emph{Generate} loop.\end{proof}
\vspace{-1ex}

Proving the C program correct w.r.t. the functional program,
then proving Theorem~\ref{spec-equiv},
is simpler than directly proving the C program implements
the crypto spec.  There are other advantages too:
\textbf{mbedtls\_generate} is directly
executable, though about a million times slower than the C
program.  Successfully validating our functional program against all 240
noReseed-noPredictionResistance test vectors from
NIST~\cite{NISTDRBGtestvectors} takes about 30 mins.

\subsection{Verification of C code}
\newcommand{\PROP}{\ensuremath{\mathsf{PROP}}}
\newcommand{\LOCAL}{\ensuremath{\mathsf{LOCAL}}}
\newcommand{\SEP}{\ensuremath{\mathsf{SEP}}} The Verified Software
Toolchain's program logic, called \emph{Verifiable C}, is a
higher-order impredicative separation logic~\cite{appel14:plcc}.
That is, it is a Hoare
logic that can handle advanced features such as pointer data
structures, function-pointers, quantification over predicates,
recursive predicates, and so on.  Its judgments take the shape of a
Hoare triple $\{ \mathit{pre}(z) \} c \{ \mathit{post}(z) \}$, where
assertions $\mathit{pre}(.)$ and $\mathit{post}(.)$ are predicates on
states. Variable $z$ is implicitly universally quantified and is used
to express functional correctness properties in which postconditions
refer to (some aspects of) preconditions. Our proof automation
operates over assertions of the form
$$\PROP\; P\; \LOCAL\; L\; \SEP\;R$$
where $P$ is a purely propositional term, $L$ describes the content
of local and global variables, and $R$ describes the content of the
heap, as a list of formulae that are implicitly interpreted to refer
to non-overlapping memory regions.

Figure~\ref{fig:randomSpec} shows a specification of
\mbox{\li|mbedtls_hmac_drbg_random|.} It means, suppose the
function is called with argument \lstinline{p_rng} holding a concrete
representation of some abstract DRBG state $I$, and \lstinline{output}
pointing to a (separate) memory region of length $0\le n \le 1024$,
specified in \lstinline{out_len}. Suppose further that requesting $n$ bytes from $I$ and entropy stream $s$ succeeds according to the
\emph{functional} specification program \lstinline{mbedtls_generate}, yielding
\lstinline{(bytes,ss, F)}. Then, executing the body of the C function
is \emph{safe}: it will not perform memory loads or stores outside the
region specified by the precondition's $\SEP$ clause, the function's
stack frame, and potentially locally allocated memory; it will not
experience division-by-zero or other runtime errors. If terminating,
the function will yield return value 0, and a state in which
\lstinline{ctx} indeed holds abstract state $F$, \lstinline{output}
indeed holds \lstinline{bytes}, and the entropy stream is updated as
specified.

\begin{figure}
\begin{lstlisting}
Definition drbg_random_abs_spec :=
DECLARE _mbedtls_hmac_drbg_random
 WITH output: val, n: Z, ctx: val, 
      I, F: hmac256drbgabs, kv: val, 
      s, ss: ENTROPY.stream, bytes: list Z
 PRE [_p_rng OF tptr tvoid, _output OF tptr tuchar,
      _out_len OF tuint ]
  PROP (0 <= n <= 1024;
        mbedtls_generate s I n = Some (bytes, ss, F)
  LOCAL (temp _p_rng ctx; temp _output output;
         temp _out_len (Vint (Int.repr n));
         gvar sha._K256 kv)
  SEP (data_at_ Tsh (tarray tuchar n) output;
       AREP kv I ctx; Stream s)
 POST [ tint ]   
  PROP () 
  LOCAL (temp ret_temp (Vint Int.zero))
  SEP (data_at Tsh (tarray tuchar n)
                   (map Vint (map Int.repr bytes)) output;
       AREP kv F ctx; Stream ss).
\end{lstlisting}
\caption{\label{fig:randomSpec}
  \mbox{VST specification of  $\mathtt{mbedtls\_hmac\_drbg\_random}$}}
\end{figure}

The use of representation predicate \lstinline{AREP} to relate an
abstract DRBG state to a concrete C data structure laid out in memory
is typical in separation logic, and emphasizes that client
code should only operate on DRBG contexts using the interface
functions. Abstract states $I$ and $F$ are elements of the type
\li|hmac256drbgabs| defined above.   The definition of \lstinline{AREP}
includes numeric invariants asserting that,
e.g.,~\lstinline{reseedCounter} is bounded between 0 and
$\mathtt{Int.max\_signed}$, but is only unfolded in the proofs of the
function bodies.

The code of $\mathtt{mbedtls\_hmac\_drbg\_random}$\; invokes the more
general function $\mathtt{mbedtls\_hmac\_drbg\_random\_with\_add}$, our
specification of which indeed supports the supply of additional input,
respects the prediction resistance flag and reseeds the DRBG if
necessary. The specification shown in Figure~\ref{fig:randomSpec}
restricts input arguments so that certain error cases are avoided;
again, our formalization also includes more general specifications.

Our C implementation is based on mbedTLS \texttt{hmac\_drbg.c},
specialized to HMAC-SHA256 by instantiating C preprocessor macros.
We make minor local edits to satisfy Verifiable C's
syntactic restrictions: isolating memory loads/stores
into separate assignment statements.
The program uses a \texttt{volatile} keyword to
prevent the C compiler from deleting the
clearing of an array after its last use;
reasoning about \texttt{volatile} is not supported
in Verifiable C, so we had to remove it.
Calls to HMAC are
rerouted to the OpenSSL implementation \cite{beringer15:hmac}.

To formally prove that the C code satisfies the specification, we use
the VST-Floyd proof automation system: a collection of Coq tactics,
lemmas, and auxiliary definitions that perform forward symbolic
execution with substantial interactive help from the user, who
supplies function specifications, loop invariants, and proofs of
application-specific transformations using Coq's built-in proof
commands. Building on the verification of $\mathtt{HMAC\_SHA256}$, we
have specified and proved all core functions in \texttt{hmac\_drbg.c},
including the flag update functions.

% stats about proof size are now in ``Proof Effort'' section

%\begin{tabular}{|c|c|c|c|}
% Function & AST size (\# nodes) & Proof size (loc) & Time (secs) \\ \hline
% A & 23 & 120 & 19 \\ \hline
% \end{tabular}

%%% Local Variables:
%%% mode: latex
%%% TeX-master: "header"
%%% End:

%% file: appendix.tex
\newpage
\section*{Appendix:  Debian/OpenSSL fiasco}

\label{debianfiasco}

In 2006, in order to eliminate a Purify warning,
a developer at Debian removed a line of code from
the randomness-initializer of OpenSSL as shipped with Debian
\cite{debianflaw}.
This reduced the amount of true entropy for initializing SSH keys to approximately 15 bits, or (in some realistic scenarios) 0 bits.

Cox \cite{cox08:debianfiasco}  has an excellent analysis
of the technical problem \emph{and the social processes}
that led to this failure.  We have no library in Coq for
proving things about social processes, but we can say:

\begin{enumerate}
\item The all-too-clever C code in the OpenSSL RAND\_poll
  and RAND\_add tolerates the use of a partially uninitialized array.
  Cox explains why (a) it's not strictly speaking \emph{wrong} to
  do this and (b) it's still a terrible idea.
  The \emph{Verifiable C} proof system would reject this program.

\item The Debian developer asked about the safety of this change
  on the official OpenSSL developer's mailing list, and got no response.
  Verifiable C cannot help with that.

\item Presumably a (hypothetical) verified implementation distributed
  by OpenSSL would be marked as such (in comments) and distributed
  with its open-source proof, and not so casually modified by the
  Debian developer.

\item Suppose there had been a different insufficient-entropy bug, one not involving
  uninitialized variables.  In such a case, it would not be automatically
  rejected by Verifiable C.  Instead, either (a) the C code could not
  be proved to implement the functional model,
  or (b) the functional model could not be proved to provide sufficient
  entropy at initialization time.  In either case, the bug would be detected.

\item Suppose the functional model and C code were both proved correct;
  the correct program asks an external device for $n$ bits of ``true'' entropy,
  and instead gets $n$ bits with almost no entropy \cite{becker2013stealthy}.
  Our proofs cannot help here.
\end{enumerate}

In summary: Yes, proofs \kye{of what properties?} would probably have prevented the Debian/OpenSSL fiasco.
\awa{Lennart had some remarks about this sentence...}

%%% Local Variables:
%%% mode: latex
%%% TeX-master: "header"
%%% End:

%% file: header.bbl
%%% -*-BibTeX-*-
%%% Do NOT edit. File created by BibTeX with style
%%% ACM-Reference-Format-Journals [18-Jan-2012].

\begin{thebibliography}{00}

%%% ====================================================================
%%% NOTE TO THE USER: you can override these defaults by providing
%%% customized versions of any of these macros before the \bibliography
%%% command.  Each of them MUST provide its own final punctuation,
%%% except for \shownote{}, \showDOI{}, and \showURL{}.  The latter two
%%% do not use final punctuation, in order to avoid confusing it with
%%% the Web address.
%%%
%%% To suppress output of a particular field, define its macro to expand
%%% to an empty string, or better, \unskip, like this:
%%%
%%% \newcommand{\showDOI}[1]{\unskip}   % LaTeX syntax
%%%
%%% \def \showDOI #1{\unskip}           % plain TeX syntax
%%%
%%% ====================================================================

\ifx \showCODEN    \undefined \def \showCODEN     #1{\unskip}     \fi
\ifx \showDOI      \undefined \def \showDOI       #1{{\tt DOI:}\penalty0{#1}\ }
  \fi
\ifx \showISBNx    \undefined \def \showISBNx     #1{\unskip}     \fi
\ifx \showISBNxiii \undefined \def \showISBNxiii  #1{\unskip}     \fi
\ifx \showISSN     \undefined \def \showISSN      #1{\unskip}     \fi
\ifx \showLCCN     \undefined \def \showLCCN      #1{\unskip}     \fi
\ifx \shownote     \undefined \def \shownote      #1{#1}          \fi
\ifx \showarticletitle \undefined \def \showarticletitle #1{#1}   \fi
\ifx \showURL      \undefined \def \showURL       #1{#1}          \fi
% The following commands are used for tagged output and should be
% invisible to TeX
\providecommand\bibfield[2]{#2}
\providecommand\bibinfo[2]{#2}
\providecommand\natexlab[1]{#1}
\providecommand\showeprint[2][]{arXiv:#2}

\bibitem[\protect\citeauthoryear{Affeldt, Nowak, and Yamada}{Affeldt
  et~al\mbox{.}}{2012}]%
        {affeldt2012certifying}
\bibfield{author}{\bibinfo{person}{Reynald Affeldt}, \bibinfo{person}{David
  Nowak}, {and} \bibinfo{person}{Kiyoshi Yamada}.}
  \bibinfo{year}{2012}\natexlab{}.
\newblock \showarticletitle{Certifying assembly with formal security proofs:
  The case of {BBS}}.
\newblock \bibinfo{journal}{{\em Science of Computer Programming\/}}
  \bibinfo{volume}{77}, \bibinfo{number}{10} (\bibinfo{year}{2012}),
  \bibinfo{pages}{1058--1074}.
\newblock


\bibitem[\protect\citeauthoryear{Almeida, Barbosa, Barthe, and
  Dupressoir}{Almeida et~al\mbox{.}}{2015}]%
        {DBLP:journals/iacr/AlmeidaBBD15}
\bibfield{author}{\bibinfo{person}{Jos{\'{e}}~Bacelar Almeida},
  \bibinfo{person}{Manuel Barbosa}, \bibinfo{person}{Gilles Barthe}, {and}
  \bibinfo{person}{Fran{\c{c}}ois Dupressoir}.}
  \bibinfo{year}{2015}\natexlab{}.
\newblock \showarticletitle{Verifiable side-channel security of cryptographic
  implementations: constant-time {MEE-CBC}}.
\newblock \bibinfo{journal}{{\em {IACR} Cryptology ePrint Archive\/}}
  \bibinfo{volume}{2015} (\bibinfo{year}{2015}), \bibinfo{pages}{1241}.
\newblock
\showURL{%
\url{http://eprint.iacr.org/2015/1241}}


\bibitem[\protect\citeauthoryear{Almeida, Barbosa, Barthe, Dupressoir, and
  Emmi}{Almeida et~al\mbox{.}}{2016}]%
        {DBLP:conf/uss/AlmeidaBBDE16}
\bibfield{author}{\bibinfo{person}{Jos{\'{e}}~Bacelar Almeida},
  \bibinfo{person}{Manuel Barbosa}, \bibinfo{person}{Gilles Barthe},
  \bibinfo{person}{Fran{\c{c}}ois Dupressoir}, {and} \bibinfo{person}{Michael
  Emmi}.} \bibinfo{year}{2016}\natexlab{}.
\newblock \showarticletitle{Verifying Constant-Time Implementations}. In
  \bibinfo{booktitle}{{\em 25th {USENIX} Security Symposium, {USENIX} Security
  16}}, \bibfield{editor}{\bibinfo{person}{Thorsten Holz} {and}
  \bibinfo{person}{Stefan Savage}} (Eds.). \bibinfo{publisher}{{USENIX}
  Association}, \bibinfo{pages}{53--70}.
\newblock
\showURL{%
\url{https://www.usenix.org/conference/usenixsecurity16/technical-sessions/presentation/almeida}}


\bibitem[\protect\citeauthoryear{Appel}{Appel}{2015}]%
        {appel15:sha}
\bibfield{author}{\bibinfo{person}{Andrew~W. Appel}.}
  \bibinfo{year}{2015}\natexlab{}.
\newblock \showarticletitle{Verification of a Cryptographic Primitive:
  {SHA}-256}.
\newblock \bibinfo{journal}{{\em ACM Trans. on Programming Languages and
  Systems\/}} \bibinfo{volume}{37}, \bibinfo{number}{2} (\bibinfo{date}{April}
  \bibinfo{year}{2015}), \bibinfo{pages}{7:1--7:31}.
\newblock


\bibitem[\protect\citeauthoryear{Appel, Dockins, Hobor, Beringer, Dodds,
  Stewart, Blazy, and Leroy}{Appel et~al\mbox{.}}{2014}]%
        {appel14:plcc}
\bibfield{author}{\bibinfo{person}{Andrew~W. Appel}, \bibinfo{person}{Robert
  Dockins}, \bibinfo{person}{Aquinas Hobor}, \bibinfo{person}{Lennart
  Beringer}, \bibinfo{person}{Josiah Dodds}, \bibinfo{person}{Gordon Stewart},
  \bibinfo{person}{Sandrine Blazy}, {and} \bibinfo{person}{Xavier Leroy}.}
  \bibinfo{year}{2014}\natexlab{}.
\newblock \bibinfo{booktitle}{{\em Program Logics for Certified Compilers}}.
\newblock \bibinfo{publisher}{Cambridge University Press},
  \bibinfo{address}{New York}.
\newblock


\bibitem[\protect\citeauthoryear{Barker and Kelsey}{Barker and Kelsey}{2012}]%
        {sp800-90A}
\bibfield{author}{\bibinfo{person}{Elaine Barker} {and} \bibinfo{person}{John
  Kelsey}.} \bibinfo{year}{2012}\natexlab{}.
\newblock \bibinfo{booktitle}{{\em Special Publication 800-90A: Recommendation
  for Random Number Generation Using Deterministic Random Bit Generators}}.
\newblock \bibinfo{type}{{T}echnical {R}eport} 800-90A.
  \bibinfo{institution}{National Institute of Standards and Technology}.
\newblock
\showURL{%
\url{http://csrc.nist.gov/publications/nistpubs/800-90A/SP800-90A.pdf}}


\bibitem[\protect\citeauthoryear{Barthe, Betarte, Campo, Luna, and
  Pichardie}{Barthe et~al\mbox{.}}{2014}]%
        {DBLP:conf/ccs/BartheBCLP14}
\bibfield{author}{\bibinfo{person}{Gilles Barthe}, \bibinfo{person}{Gustavo
  Betarte}, \bibinfo{person}{Juan~Diego Campo}, \bibinfo{person}{Carlos~Daniel
  Luna}, {and} \bibinfo{person}{David Pichardie}.}
  \bibinfo{year}{2014}\natexlab{}.
\newblock \showarticletitle{System-level Non-interference for Constant-time
  Cryptography}. In \bibinfo{booktitle}{{\em Proceedings of the 2014 {ACM}
  {SIGSAC} Conference on Computer and Communications Security}},
  \bibfield{editor}{\bibinfo{person}{Gail{-}Joon Ahn}, \bibinfo{person}{Moti
  Yung}, {and} \bibinfo{person}{Ninghui Li}} (Eds.).
  \bibinfo{publisher}{{ACM}}, \bibinfo{pages}{1267--1279}.
\newblock
\showISBNx{978-1-4503-2957-6}
\showDOI{%
\url{https://doi.org/10.1145/2660267.2660283}}


\bibitem[\protect\citeauthoryear{Becker, Regazzoni, Paar, and Burleson}{Becker
  et~al\mbox{.}}{2013}]%
        {becker2013stealthy}
\bibfield{author}{\bibinfo{person}{Georg~T. Becker}, \bibinfo{person}{Francesco
  Regazzoni}, \bibinfo{person}{Christof Paar}, {and} \bibinfo{person}{Wayne~P.
  Burleson}.} \bibinfo{year}{2013}\natexlab{}.
\newblock \showarticletitle{Stealthy Dopant-level Hardware Trojans}. In
  \bibinfo{booktitle}{{\em International Workshop on Cryptographic Hardware and
  Embedded Systems}}. \bibinfo{publisher}{Springer}, \bibinfo{address}{Berlin},
  \bibinfo{pages}{197--214}.
\newblock


\bibitem[\protect\citeauthoryear{Bellare}{Bellare}{2006}]%
        {bellare2006new}
\bibfield{author}{\bibinfo{person}{Mihir Bellare}.}
  \bibinfo{year}{2006}\natexlab{}.
\newblock \showarticletitle{New proofs for NMAC and HMAC: Security without
  collision-resistance}. In \bibinfo{booktitle}{{\em Annual International
  Cryptology Conference}}. \bibinfo{publisher}{Springer},
  \bibinfo{pages}{602--619}.
\newblock


\bibitem[\protect\citeauthoryear{Bellare and Rogaway}{Bellare and
  Rogaway}{2006}]%
        {bellare2006security}
\bibfield{author}{\bibinfo{person}{Mihir Bellare} {and}
  \bibinfo{person}{Phillip Rogaway}.} \bibinfo{year}{2006}\natexlab{}.
\newblock \showarticletitle{The security of triple encryption and a framework
  for code-based game-playing proofs}. In \bibinfo{booktitle}{{\em Annual
  International Conference on the Theory and Applications of Cryptographic
  Techniques}}. \bibinfo{publisher}{Springer}, \bibinfo{pages}{409--426}.
\newblock


\bibitem[\protect\citeauthoryear{Beringer, Petcher, Ye, and Appel}{Beringer
  et~al\mbox{.}}{2015}]%
        {beringer15:hmac}
\bibfield{author}{\bibinfo{person}{Lennart Beringer}, \bibinfo{person}{Adam
  Petcher}, \bibinfo{person}{Katherine~Q. Ye}, {and} \bibinfo{person}{Andrew~W.
  Appel}.} \bibinfo{year}{2015}\natexlab{}.
\newblock \showarticletitle{Verified Correctness and Security of {OpenSSL}
  {HMAC}}. In \bibinfo{booktitle}{{\em 24th USENIX Security Symposium}}.
  \bibinfo{publisher}{USENIX Assocation}, \bibinfo{pages}{207--221}.
\newblock


\bibitem[\protect\citeauthoryear{Bernstein, Chang, Cheng, Chou, Heninger,
  Lange, and van Someren}{Bernstein et~al\mbox{.}}{2013}]%
        {Bernstein13}
\bibfield{author}{\bibinfo{person}{Daniel~J. Bernstein},
  \bibinfo{person}{Yun-An Chang}, \bibinfo{person}{Chen-Mou Cheng},
  \bibinfo{person}{Li-Ping Chou}, \bibinfo{person}{Nadia Heninger},
  \bibinfo{person}{Tanja Lange}, {and} \bibinfo{person}{Nicko van Someren}.}
  \bibinfo{year}{2013}\natexlab{}.
\newblock \showarticletitle{Factoring {RSA} Keys from Certified Smart Cards:
  {Coppersmith} in the Wild}. In \bibinfo{booktitle}{{\em ASIACRYPT 2013: 19th
  International Conference on the Theory and Application of Cryptology and
  Information Security, Proceedings Part II}}. \bibinfo{publisher}{Springer},
  \bibinfo{address}{Berlin}, \bibinfo{pages}{341--360}.
\newblock
\showISBNx{978-3-642-42045-0}
\showDOI{%
\url{https://doi.org/10.1007/978-3-642-42045-0_18}}


\bibitem[\protect\citeauthoryear{Bernstein, Lange, and Niederhagen}{Bernstein
  et~al\mbox{.}}{2016}]%
        {DBLP:conf/birthday/BernsteinLN16}
\bibfield{author}{\bibinfo{person}{Daniel~J. Bernstein}, \bibinfo{person}{Tanja
  Lange}, {and} \bibinfo{person}{Ruben Niederhagen}.}
  \bibinfo{year}{2016}\natexlab{}.
\newblock \showarticletitle{Dual {EC:} {A} Standardized Back Door}. In
  \bibinfo{booktitle}{{\em The New Codebreakers - Essays Dedicated to David
  Kahn on the Occasion of His 85th Birthday}} {\em (\bibinfo{series}{Lecture
  Notes in Computer Science})}, \bibfield{editor}{\bibinfo{person}{Peter Y.~A.
  Ryan}, \bibinfo{person}{David Naccache}, {and}
  \bibinfo{person}{Jean{-}Jacques Quisquater}} (Eds.),
  Vol.~\bibinfo{volume}{9100}. \bibinfo{publisher}{Springer},
  \bibinfo{pages}{256--281}.
\newblock
\showISBNx{978-3-662-49300-7}
\showDOI{%
\url{https://doi.org/10.1007/978-3-662-49301-4_17}}


\bibitem[\protect\citeauthoryear{Bhargavan, Fournet, Kohlweiss, Pironti, and
  Strub}{Bhargavan et~al\mbox{.}}{2013}]%
        {bhargavan2013implementing}
\bibfield{author}{\bibinfo{person}{Karthikeyan Bhargavan},
  \bibinfo{person}{C{\'e}dric Fournet}, \bibinfo{person}{Markulf Kohlweiss},
  \bibinfo{person}{Alfredo Pironti}, {and} \bibinfo{person}{Pierre-Yves
  Strub}.} \bibinfo{year}{2013}\natexlab{}.
\newblock \showarticletitle{Implementing {TLS} with verified cryptographic
  security}. In \bibinfo{booktitle}{{\em 2013 IEEE Symposium on Security and
  Privacy}}. \bibinfo{publisher}{IEEE}, \bibinfo{pages}{445--459}.
\newblock


\bibitem[\protect\citeauthoryear{Checkoway, Maskiewicz, Garman, Fried, Cohney,
  Green, Heninger, Weinmann, Rescorla, and Shacham}{Checkoway
  et~al\mbox{.}}{2016}]%
        {Checkoway16}
\bibfield{author}{\bibinfo{person}{Stephen Checkoway}, \bibinfo{person}{Jacob
  Maskiewicz}, \bibinfo{person}{Christina Garman}, \bibinfo{person}{Joshua
  Fried}, \bibinfo{person}{Shaanan Cohney}, \bibinfo{person}{Matthew Green},
  \bibinfo{person}{Nadia Heninger}, \bibinfo{person}{Ralf-Philipp Weinmann},
  \bibinfo{person}{Eric Rescorla}, {and} \bibinfo{person}{Hovav Shacham}.}
  \bibinfo{year}{2016}\natexlab{}.
\newblock \showarticletitle{A Systematic Analysis of the {Juniper Dual EC}
  Incident}. In \bibinfo{booktitle}{{\em CCS '16: 23rd ACM Conference on
  Computer and Communications Security}}. \bibinfo{publisher}{ACM},
  \bibinfo{address}{New York, NY, USA}, \bibinfo{pages}{468--479}.
\newblock
\showISBNx{978-1-4503-4139-4}
\showDOI{%
\url{https://doi.org/10.1145/2976749.2978395}}


\bibitem[\protect\citeauthoryear{Checkoway, Niederhagen, Everspaugh, Green,
  Lange, Ristenpart, Bernstein, Maskiewicz, Shacham, and Fredrikson}{Checkoway
  et~al\mbox{.}}{2014}]%
        {Checkoway14}
\bibfield{author}{\bibinfo{person}{Stephen Checkoway}, \bibinfo{person}{Ruben
  Niederhagen}, \bibinfo{person}{Adam Everspaugh}, \bibinfo{person}{Matthew
  Green}, \bibinfo{person}{Tanja Lange}, \bibinfo{person}{Thomas Ristenpart},
  \bibinfo{person}{Daniel~J. Bernstein}, \bibinfo{person}{Jake Maskiewicz},
  \bibinfo{person}{Hovav Shacham}, {and} \bibinfo{person}{Matthew Fredrikson}.}
  \bibinfo{year}{2014}\natexlab{}.
\newblock \showarticletitle{On the Practical Exploitability of {Dual EC} in
  {TLS} Implementations}. In \bibinfo{booktitle}{{\em Usenix Security '14}}.
  \bibinfo{publisher}{USENIX Association}, \bibinfo{address}{San Diego, CA},
  \bibinfo{pages}{319--335}.
\newblock
\showISBNx{978-1-931971-15-7}
\showURL{%
\url{https://www.usenix.org/conference/usenixsecurity14/technical-sessions/presentation/checkoway}}


\bibitem[\protect\citeauthoryear{Cox}{Cox}{2008}]%
        {cox08:debianfiasco}
\bibfield{author}{\bibinfo{person}{Russ Cox}.} \bibinfo{year}{2008}\natexlab{}.
\newblock \bibinfo{title}{Lessons from the Debian/OpenSSL Fiasco}.
  (\bibinfo{date}{21 May} \bibinfo{year}{2008}).
\newblock
\showURL{%
\url{https://research.swtch.com/openssl}}


\bibitem[\protect\citeauthoryear{Dodis, Ganesh, Golovnev, Juels, and
  Ristenpart}{Dodis et~al\mbox{.}}{2015}]%
        {dodis2015formal}
\bibfield{author}{\bibinfo{person}{Yevgeniy Dodis}, \bibinfo{person}{Chaya
  Ganesh}, \bibinfo{person}{Alexander Golovnev}, \bibinfo{person}{Ari Juels},
  {and} \bibinfo{person}{Thomas Ristenpart}.} \bibinfo{year}{2015}\natexlab{}.
\newblock \showarticletitle{A formal treatment of backdoored pseudorandom
  generators}. In \bibinfo{booktitle}{{\em EUROCRYPT (1)}}.
  \bibinfo{pages}{101--126}.
\newblock


\bibitem[\protect\citeauthoryear{Dodis, Gennaro, H{\aa}stad, Krawczyk, and
  Rabin}{Dodis et~al\mbox{.}}{2004}]%
        {dodis2004randomness}
\bibfield{author}{\bibinfo{person}{Yevgeniy Dodis}, \bibinfo{person}{Rosario
  Gennaro}, \bibinfo{person}{Johan H{\aa}stad}, \bibinfo{person}{Hugo
  Krawczyk}, {and} \bibinfo{person}{Tal Rabin}.}
  \bibinfo{year}{2004}\natexlab{}.
\newblock \showarticletitle{Randomness extraction and key derivation using the
  {CBC}, cascade and {HMAC} modes}. In \bibinfo{booktitle}{{\em Annual
  International Cryptology Conference}}. \bibinfo{publisher}{Springer},
  \bibinfo{pages}{494--510}.
\newblock


\bibitem[\protect\citeauthoryear{D\"{o}rre and Klebanov}{D\"{o}rre and
  Klebanov}{2016}]%
        {DorreKlebanov2016}
\bibfield{author}{\bibinfo{person}{Felix D\"{o}rre} {and}
  \bibinfo{person}{Vladimir Klebanov}.} \bibinfo{year}{2016}\natexlab{}.
\newblock \showarticletitle{Practical Detection of Entropy Loss in
  Pseudo-Random Number Generators}. In \bibinfo{booktitle}{{\em CCS'16: 2016
  ACM SIGSAC Conference on Computer and Communications Security}}.
  \bibinfo{publisher}{ACM}, \bibinfo{address}{New York},
  \bibinfo{pages}{678--689}.
\newblock


\bibitem[\protect\citeauthoryear{Fouque, Pointcheval, and Zimmer}{Fouque
  et~al\mbox{.}}{2008}]%
        {fouque2008hmac}
\bibfield{author}{\bibinfo{person}{Pierre-Alain Fouque}, \bibinfo{person}{David
  Pointcheval}, {and} \bibinfo{person}{S{\'e}bastien Zimmer}.}
  \bibinfo{year}{2008}\natexlab{}.
\newblock \showarticletitle{{HMAC} is a randomness extractor and applications
  to {TLS}}. In \bibinfo{booktitle}{{\em ACM symposium on Information, Computer
  and Communications Security}}. \bibinfo{publisher}{ACM},
  \bibinfo{pages}{21--32}.
\newblock


\bibitem[\protect\citeauthoryear{Goodin}{Goodin}{2013}]%
        {goodin13}
\bibfield{author}{\bibinfo{person}{Dan Goodin}.}
  \bibinfo{year}{2013}\natexlab{}.
\newblock \showarticletitle{Google confirms critical {Android} crypto flaw used
  in \$5,700 {Bitcoin} heist}.
\newblock \bibinfo{journal}{{\em Ars Technica\/}} (\bibinfo{date}{14 Aug.}
  \bibinfo{year}{2013}).
\newblock
\showURL{%
\url{https://arstechnica.com/security/2013/08/google-confirms-critical-android-crypto-flaw-used-in-5700-bitcoin-heist/}}


\bibitem[\protect\citeauthoryear{Heninger, Durumeric, Wustrow, and
  Halderman}{Heninger et~al\mbox{.}}{2012}]%
        {Heninger12}
\bibfield{author}{\bibinfo{person}{Nadia Heninger}, \bibinfo{person}{Zakir
  Durumeric}, \bibinfo{person}{Eric Wustrow}, {and} \bibinfo{person}{J.~Alex
  Halderman}.} \bibinfo{year}{2012}\natexlab{}.
\newblock \showarticletitle{Mining Your {Ps} and {Qs}: Detection of Widespread
  Weak Keys in Network Devices}. In \bibinfo{booktitle}{{\em 21st USENIX
  Security Symposium}}. \bibinfo{publisher}{USENIX Association},
  \bibinfo{pages}{205--220}.
\newblock


\bibitem[\protect\citeauthoryear{Hirose}{Hirose}{2008}]%
        {hirose2008security}
\bibfield{author}{\bibinfo{person}{Shoichi Hirose}.}
  \bibinfo{year}{2008}\natexlab{}.
\newblock \showarticletitle{Security analysis of DRBG using HMAC in NIST SP
  800-90}. In \bibinfo{booktitle}{{\em International Workshop on Information
  Security Applications}}. \bibinfo{publisher}{Springer},
  \bibinfo{pages}{278--291}.
\newblock


\bibitem[\protect\citeauthoryear{Inc.}{Inc.}{2016}]%
        {FortinetBug}
\bibfield{author}{\bibinfo{person}{Fortinet Inc.}}
  \bibinfo{year}{2016}\natexlab{}.
\newblock \bibinfo{title}{CVE-2016-8492}.
\newblock \bibinfo{howpublished}{Available at
  \url{https://fortiguard.com/advisory/FG-IR-16-067}|}.
  (\bibinfo{year}{2016}).
\newblock


\bibitem[\protect\citeauthoryear{ISO}{ISO}{2012}]%
        {iso19790}
\bibfield{author}{\bibinfo{person}{ISO}.} \bibinfo{year}{2012}\natexlab{}.
\newblock \bibinfo{title}{{ISO 19790:2012: Security requirements for
  cryptographic modules}}.
\newblock \bibinfo{howpublished}{Available at
  \url{https://www.iso.org/standard/52906.html}}.   (\bibinfo{date}{August}
  \bibinfo{year}{2012}).
\newblock


\bibitem[\protect\citeauthoryear{Krawczyk}{Krawczyk}{2010}]%
        {krawczyk2010cryptographic}
\bibfield{author}{\bibinfo{person}{Hugo Krawczyk}.}
  \bibinfo{year}{2010}\natexlab{}.
\newblock \showarticletitle{Cryptographic extraction and key derivation: The
  HKDF scheme}. In \bibinfo{booktitle}{{\em Annual Cryptology Conference}}.
  Springer, \bibinfo{pages}{631--648}.
\newblock


\bibitem[\protect\citeauthoryear{Lenstra, Hughes, Augier, Bos, Kleinjung, and
  Wachter}{Lenstra et~al\mbox{.}}{2012}]%
        {Lenstra12}
\bibfield{author}{\bibinfo{person}{Arjen~K. Lenstra}, \bibinfo{person}{James~P.
  Hughes}, \bibinfo{person}{Maxime Augier}, \bibinfo{person}{Joppe~W. Bos},
  \bibinfo{person}{Thorsten Kleinjung}, {and} \bibinfo{person}{Christophe
  Wachter}.} \bibinfo{year}{2012}\natexlab{}.
\newblock \bibinfo{title}{{Ron} was wrong, {Whit} is right}.
\newblock \bibinfo{howpublished}{Cryptology ePrint Archive, Report 2012/064}.
  (\bibinfo{year}{2012}).
\newblock
\newblock
\shownote{\url{http://eprint.iacr.org/2012/064}.}


\bibitem[\protect\citeauthoryear{Moore}{Moore}{2008}]%
        {debianflaw}
\bibfield{author}{\bibinfo{person}{H.~D. Moore}.}
  \bibinfo{year}{2008}\natexlab{}.
\newblock \bibinfo{title}{{Debian OpenSSL Flaw}}.
\newblock \bibinfo{howpublished}{Available at
  \url{https://hdm.io/tools/debian-openssl/}}.   (\bibinfo{year}{2008}).
\newblock


\bibitem[\protect\citeauthoryear{{National Institute of Standards and
  Technology}}{{National Institute of Standards and Technology}}{2014}]%
        {NISTRNGToolkit}
\bibfield{author}{\bibinfo{person}{{National Institute of Standards and
  Technology}}.} \bibinfo{year}{2014}\natexlab{}.
\newblock \bibinfo{title}{{NIST RNG Cryptographic Toolkit}}.
\newblock \bibinfo{howpublished}{Available at
  \url{http://csrc.nist.gov/groups/ST/toolkit/rng/}}.   (\bibinfo{date}{July}
  \bibinfo{year}{2014}).
\newblock


\bibitem[\protect\citeauthoryear{{National Institute of Standards and
  Technology}}{{National Institute of Standards and Technology}}{2017}]%
        {NISTDRBGtestvectors}
\bibfield{author}{\bibinfo{person}{{National Institute of Standards and
  Technology}}.} \bibinfo{year}{2017}\natexlab{}.
\newblock \bibinfo{title}{CAVP Testing: Random Number Generators}.
\newblock   (\bibinfo{year}{2017}).
\newblock
\newblock
\shownote{http://csrc.nist.gov/groups/STM/cavp/random-number-generation.html.}


\bibitem[\protect\citeauthoryear{{National Institute of Standards and
  Technology (NIST)}}{{National Institute of Standards and Technology
  (NIST)}}{2001}]%
        {fips140-2}
\bibfield{author}{\bibinfo{person}{{National Institute of Standards and
  Technology (NIST)}}.} \bibinfo{year}{2001}\natexlab{}.
\newblock \bibinfo{title}{{FIPS 140-2: Security Requirements for Cryptographic
  Modules}}.
\newblock \bibinfo{howpublished}{Available at
  \url{http://nvlpubs.nist.gov/nistpubs/FIPS/NIST.FIPS.140-2.pdf}}.
  (\bibinfo{date}{May} \bibinfo{year}{2001}).
\newblock


\bibitem[\protect\citeauthoryear{Perlroth, Larson, and Shane}{Perlroth
  et~al\mbox{.}}{2013}]%
        {NYT_DualEC_2013}
\bibfield{author}{\bibinfo{person}{Nicole Perlroth}, \bibinfo{person}{Jeff
  Larson}, {and} \bibinfo{person}{Scott Shane}.}
  \bibinfo{year}{2013}\natexlab{}.
\newblock \showarticletitle{{N.S.A.} Able to Foil Basic Safeguards of Privacy
  on Web}.
\newblock \bibinfo{journal}{{\em The New York Times\/}} (\bibinfo{date}{6
  Sept.} \bibinfo{year}{2013}).
\newblock
\showURL{%
\url{http://www.nytimes.com/2013/09/06/us/nsa-foils-much-internet-encryption.html}}


\bibitem[\protect\citeauthoryear{Petcher}{Petcher}{2015}]%
        {petcher15:phd}
\bibfield{author}{\bibinfo{person}{Adam Petcher}.}
  \bibinfo{year}{2015}\natexlab{}.
\newblock {\em \bibinfo{title}{A Foundational Proof Framework for
  Cryptography}}.
\newblock \bibinfo{thesistype}{Ph.D. Dissertation}. \bibinfo{school}{Harvard
  University}.
\newblock


\bibitem[\protect\citeauthoryear{Petcher and Morrisett}{Petcher and
  Morrisett}{2015}]%
        {petcher15:post}
\bibfield{author}{\bibinfo{person}{Adam Petcher} {and} \bibinfo{person}{Greg
  Morrisett}.} \bibinfo{year}{2015}\natexlab{}.
\newblock \showarticletitle{The Foundational Cryptography Framework}. In
  \bibinfo{booktitle}{{\em 4th International Conference on Principles of
  Security and Trust (POST)}}. \bibinfo{publisher}{Springer LNCS 9036},
  \bibinfo{address}{Berlin}, \bibinfo{pages}{53--72}.
\newblock


\bibitem[\protect\citeauthoryear{Shumow and Ferguson}{Shumow and
  Ferguson}{2007}]%
        {Shumow07}
\bibfield{author}{\bibinfo{person}{Dan Shumow} {and} \bibinfo{person}{Niels
  Ferguson}.} \bibinfo{year}{2007}\natexlab{}.
\newblock \bibinfo{title}{On the Possibility of a Back Door in the {NIST
  SP800-90} {Dual Ec Prng}}.
\newblock \bibinfo{howpublished}{CRYPTO 2007 Rump Session}.
  (\bibinfo{date}{Aug.} \bibinfo{year}{2007}).
\newblock
\showURL{%
\url{http://rump2007.cr.yp.to/15-shumow.pdf}}


\bibitem[\protect\citeauthoryear{Sidorenko and Schoenmakers}{Sidorenko and
  Schoenmakers}{2005}]%
        {sidorenko05}
\bibfield{author}{\bibinfo{person}{Andrey Sidorenko} {and}
  \bibinfo{person}{Berry Schoenmakers}.} \bibinfo{year}{2005}\natexlab{}.
\newblock \showarticletitle{Concrete Security of the {Blum-Blum-Shub}
  Pseudorandom Generator}. In \bibinfo{booktitle}{{\em Cryptography and Coding:
  10th IMA International Conference}}. \bibinfo{publisher}{Springer LNCS 3796},
  \bibinfo{address}{Berlin}, \bibinfo{pages}{355--375}.
\newblock


\bibitem[\protect\citeauthoryear{Thompson}{Thompson}{1984}]%
        {thompson1984reflections}
\bibfield{author}{\bibinfo{person}{Ken Thompson}.}
  \bibinfo{year}{1984}\natexlab{}.
\newblock \showarticletitle{Reflections on trusting trust}.
\newblock \bibinfo{journal}{{\it Commun. ACM}} \bibinfo{volume}{27},
  \bibinfo{number}{8} (\bibinfo{year}{1984}), \bibinfo{pages}{761--763}.
\newblock


\bibitem[\protect\citeauthoryear{Zetter}{Zetter}{2016}]%
        {Wired_Juniper_2015}
\bibfield{author}{\bibinfo{person}{Kim Zetter}.}
  \bibinfo{year}{2016}\natexlab{}.
\newblock \showarticletitle{{New Discovery Around Juniper Backdoor Raises More
  Questions About the Company}}.
\newblock \bibinfo{journal}{{\em {WIRED.com}\/}} (\bibinfo{date}{Jan.}
  \bibinfo{year}{2016}).
\newblock
\showURL{%
\url{https://www.wired.com/2016/01/new-discovery-around-juniper-backdoor-raises-more-questions-about-the-company/}}


\end{thebibliography}
